\documentclass[a4paper,UKenglish]{lipics}
 
\usepackage{microtype}
\usepackage{amssymb}
\usepackage{amsmath,amsthm}
\usepackage{graphicx}
\usepackage{color,wrapfig,sidecap}
\usepackage{enumitem}
\usepackage{hyperref}
\usepackage{pgf,tikz}
\usetikzlibrary{arrows}
\usetikzlibrary{snakes}
\usepackage{footnote}
\usepackage{algorithm,algorithmic}

\theoremstyle{plain}
\newtheorem{claim}[theorem]{Claim}

\newtheorem{conjecture}[theorem]{Conjecture}
\newcommand{\pw}{\mathrm{pw}}

\newcommand{\poly}{\mathrm{poly}}

\newcommand{\tw}{\mathrm{tw}}

\newcommand{\MST}{\ensuremath{\mathrm{MST}}}

\newif\ifFull 
\Fullfalse
\title{Light spanners for bounded treewidth graphs imply light spanners for $H$-minor-free graphs\footnote{This material is based upon work supported by the National Science Foundation under Grant No.\ CCF-1252833.}}
\titlerunning{Light spanners bounded treewidth graphs and $H$-minor-free graphs} 

\author{Glencora Borradaile}
\author{Hung Le}
\affil{Department of Electrical Engineering and Computer Science\\ 
	Oregon State University, USA\\
  \texttt{glencora@eecs.orst.edu, lehu@onid.oregonstate.edu}}
\authorrunning{G. Borradaile and H. Le} 

\Copyright{Glencora Borradaile and Hung Le}


\serieslogo{}
\volumeinfo
  {Billy Editor and Bill Editors}
  {2}
  {Conference title on which this volume is based on}
  {1}
  {1}
  {1}
\EventShortName{Arxiv}
\DOI{10.4230/LIPIcs.xxx.yyy.p}

\begin{document}

\maketitle

\begin{abstract}
  Grigni and Hung~\cite{GH12} conjectured that H-minor-free graphs have $(1+\epsilon)$-spanners that are light, that is, of weight $g(|H|,\epsilon)$ times the weight of the minimum spanning tree for some function $g$. This conjecture implies the {\em efficient} polynomial-time approximation scheme (PTAS) of the traveling salesperson problem in $H$-minor free graphs; that is, a PTAS whose running time is of the form $2^{f(\epsilon)}n^{O(1)}$ for some function $f$. The state of the art PTAS for TSP in H-minor-free-graphs has running time $n^{1/\poly(\epsilon)}$. We take a further step toward proving this conjecture by showing that if the bounded treewidth graphs have light greedy spanners, then the conjecture is true. We also prove that the greedy spanner of a bounded pathwidth graph is light and discuss the possibility of extending our proof to bounded treewidth graphs. 
 \end{abstract}
\section{Introduction}
\ifFull
Spanners are used to approximately preserve distances in a compact way. Chew~\cite{Chew89} introduced spanners in the geometric setting and presented an application of spanners in motion planning. Since then, spanners have been used in various areas of algorithm design: metric space searching~\cite{NPC02}, geometric approximation algorithms~\cite{RS98,BE12} and distance oracles~\cite{GLNS02} to name a few. Alth\"ofer et al.~\cite{ADDJS93} introduced greedy spanners (defined below) for planar and geometric graphs; greedy spanners are of high quality (formalized below)  compared to other types of spanners, as partly confirmed by extensive experiments by Farshi and Gudmundsson~\cite{FG10} on random geometric graphs. In designing polynomial time approximation scheme (PTAS), the planar greedy spanner was used by Klein~\cite{Klein05b} to give a PTAS for TSP on planar graphs.
\else \fi

Spanners are used to approximately preserve distances in a compact way. In this work, we focus on spanners that preserve distances within a $(1+\epsilon)$ factor (for a fixed $\epsilon < 1$) and measure quality in terms of the spanner's weight compared to the minimum spanning tree (the {\em lightness}).  Formally, given an edge-weighted graph $G$, we wish to find a spanning subgraph $S$ of $G$ such that\ifFull:
\[d_S(x,y) \le (1+\epsilon)\cdot d_G(x,y)\ \forall x,y \in V(G)\]
\[w(S) = L(\epsilon)\cdot w(\MST(G))\]
\else\footnote{We use standard graph terminology and notation; we revisit notation necessary for some proofs in Appendix~\ref{app:notation}.} $d_S(x,y) \le (1+\epsilon)\cdot d_G(x,y)\ \forall x,y \in V(G)$ and $w(S) = L(\epsilon)\cdot w(\MST(G))$
\fi
where the lightness, $L(\epsilon)$, is a function that depends only on $\epsilon$\ifFull, $d_\cdot(\cdot,\cdot)$ denotes the weight of the shortest path between the given vertices in the given graph, $w(\cdot)$ denotes the sum of the weights of the edges in the given set of edges and $\MST(\cdot)$ is the minimum spanning tree of the given graph.\else.
\fi

We focus on the {\em greedy $(1+\epsilon)$-spanner}, the spanner that is constructed by adding edges by increasing weight while doing so decreases the distance between their endpoints by a $1+\epsilon$ factor.
Alth\"ofer et al.\ showed that the greedy spanner has lightness $L(\epsilon) = O(1/\epsilon)$ for planar graphs (and also gave lightness bounds that depend on $n$ for general graphs)~\cite{ADDJS93}.  \ifFull
Since then, similar results have been shown for other minor-closed families of graphs:
\begin{center}
    \begin{tabular}{ c | c | c }
      \multicolumn{3}{c}{Lightness of the Greedy $(1+\epsilon)$-Spanner}\\\hline
    Lightness, $L(\epsilon)$  & Class of graphs & Reference \\\hline
    $O(1/\epsilon)$ & planar graphs & Alth\"ofer et al.~\cite{ADDJS93} \\ 
    $O(1/\epsilon)$ & bounded genus graphs & Grigni~\cite{Grigni00} \\ 
    $O(|H|\sqrt{\log |H|}\log n/\epsilon)$ & $H$-minor-free graphs & Grigni and Sissokho~\cite{GS02} \\ 
        $O(\pw^2/\epsilon)$ & bounded-pathwidth graphs & Theorem~\ref{thm:tw-spanner}\\\hline
    \end{tabular}
\end{center}
\else The same lightness bound holds for bounded genus graphs, as showed by Grigni~\cite{Grigni00}. However, the best lightness bound, which was shown by Grigni and Sissokho~\cite{GS02}, for $H$-minor-free graphs is $(|H|\sqrt{\log |H|}\log n/\epsilon)$.  
\fi

In this work, we investigate the possibility of removing the dependence on $n$ from the lightness for  $H$-minor-free graphs by focusing on the following conjecture of Grigni and Hung~\cite{GH12}:
\begin{conjecture}\label{conj}
  $H$-minor-free graphs have $(1+\epsilon)$-spanners with lightness that depends on $|H|$ and $\epsilon$ only.
\end{conjecture}
If this conjecture is true, it would, among other things, imply that TSP admits an {\em efficient} PTAS for $H$-minor free graph, that is, a PTAS whose running time is of the form $2^{f(\epsilon)}n^{O(1)}$ for some function $f$, improving on the existing PTAS with running time $n^{1/\poly(\epsilon)}$ via the framework of Demaine, Hajiaghayi and Kawarabayashi~\cite{DHK11} and the spanner of Grigni and Sissokho~\cite{GS02}.  
We make progress towards proving Conjecture~\ref{conj} by reducing the heart of the problem to the simpler graph class of bounded treewidth\footnote{Formal definitions of pathwidth and treewidth are given later in this paper.} graphs:
\begin{theorem}\label{thm:h-minor-spanner}
If the greedy $(1+\epsilon)$-spanner of a graph of treewidth $\tw$  has lightness that depends on $\tw$ and $\epsilon$ only, then the greedy $(1+\epsilon)$-spanner of an $H$-minor-free graph has lightness that depends on $|H|$ and $\epsilon$ only.
\end{theorem}

Grigni and Hung gave a construction of a $(1+\epsilon)$-spanner for graphs of pathwidth $\pw$ with lightness $O(\pw^3/\epsilon)$~\cite{GH12}; however, their construction is not greedy. Rather than considering the edges by increasing order of weight, they constructed a {\em monotone spanning tree} and greedily added edges to the monotone tree. They also argued that such a spanning tree is unlikely to exist for bounded treewidth graphs, giving little hope on two fronts (the different spanner construction as well as a construction that is unlikely to generalize to graphs of bounded treewidth) that Theorem~\ref{thm:h-minor-spanner} will lead to proving Conjecture~\ref{conj} via Grigni and Hung's work.   In this paper we improve Grigni and Hung's for bounded pathwidth graphs and do so by arguing lightness for the standard greedy algorithm, removing the limitations of Theorem~\ref{thm:h-minor-spanner} as a stepping stone to Conjecture~\ref{conj}.  In Section~\ref{sc:bounded-pw}, we prove: 

\begin{theorem}\label{thm:tw-spanner}
  The greedy $(1+\epsilon)$-spanner for a graph $G$ of pathwidth $\pw$ has lightness $O(\pw^2/\epsilon)$.
\end{theorem}
\noindent While our proof does not immediately extend to graphs of bounded treewidth, the techniques are not as specific to path decompositions as Grigni and Hung's  monotone-spanning-tree technique is, and thus gives more hope for proving Conjecture~\ref{conj}.

While it may seem like a limitation in proving Conjecture~\ref{conj} that we must show that a particular construction of the $(1+\epsilon)$-spanner (namely the greedy construction) is light for bounded treewidth graphs, Filtser and Solomon (Theorem 4~\cite{FS16}) showed that if an edge-weighted graph has light spanner, its greedy spanner is also light. 

\ifFull
\subsection{Spanners for PTASes}
Spanners are used algorithmically.  We point out one implication that light spanners for $H$-minor free graphs would have in regards to the traveling salesperson problem (TSP).
A polynomial-time approximation scheme (PTAS) is an algorithm that, for a fixed $\epsilon$, runs in polynomial time and has approximation ratio $1+\epsilon$.  One framework for designing PTASes for TSP in minor-closed graph families starts by first taking a light $(1+\epsilon)$-spanner, contracting a fraction of the edges (whose weight is bounded because the spanner is light) to give a bounded treewidth graph and then solving the problem optimally.  A solution for the original graph is obtained by adding the contracted edges back.  Because the contracted edges are light, this does not increase the cost of the solution by much.  Because the spanner preserves distances within a $1+\epsilon$ factor, the error introduced is at most $\epsilon$ factor.

This framework was introduced by Klein~\cite{Klein05b} and inspired by PTASes for problems like independent set and vertex cover, that do not require the spanner step, due to Baker~\cite{Baker94}.  Initially this framework was only developed for planar graphs, but resulting in {\em efficient} PTASes; that is, PTASes whose running time is of the form $2^{f(\epsilon)}n^{O(1)}$ for some function $f$.  Klein's EPTAS for TSP had a running time of $2^{O(1)/{\epsilon^2}}n$; this EPTAS was generalized to bounded genus graphs by Borradaile, Demain and Tazari~\cite{BDT12} with a slight increase in the running time dependence on $\epsilon$.

Demaine, Hajiaghayi and Kawarabayashi generalized the PTAS framework from planar graphs to $H$-minor free graphs~\cite{DHK11}.  However, the best known lightness bound for the greedy spanner of $H$-minor free graphs at the time was $O(|H|\sqrt{\log |H|}\log n/\epsilon)$, due to Grigni and Sissokho~\cite{GS02}.  Since the lightness of the spanner depends logarithmically on $n$ and the PTAS framework results in an algorithm that is exponential in the lightness, this implied a PTAS with running time $n^{1/\poly(\epsilon)}$.  

However, if the $H$-minor free spanner has lightness $p(|H|,\epsilon)\cdot w(MST)$, as would be the case if bounded treewidth graphs have lightness $g(\tw,\epsilon)\cdot w(MST)$ (via Theorem~\ref{thm:h-minor-spanner}), the PTAS using the Demaine, Hajiaghayi and Kawarabayashi framework would be efficient, with running time $2^{O(p(|H|, \epsilon))}n^{O(1)}$.
\fi

\section{Analyzing greedy spanners}
The greedy construction for a $(1+\epsilon)$-spanner due to Alth\"ofer et al.~\cite{ADDJS93} is an extension of Kruskal's minimum spanning tree algorithm.  Start by sorting the edges by increasing weight and an empty spanner subgraph $S$; for each edge $uv$ in order, if $(1+\epsilon)w(uv) \le d_S(u,v)$, then $uv$ is added to $S$.  By observing that this is a relaxation of Kruskal's algorithm, $\MST(G) \subseteq S$.  Alth\"ofer et al.\ (Lemma 3~\cite{ADDJS93}) also showed that for any edge $e = uv$ in $S$ and any $u$-to-$v$ path $P_S(uv)$ between $u$ and $v$ in $S\setminus \{e\}$, we have:
\begin{equation} \label{eq:edge-path-ineq}
(1+\epsilon)w(e) \leq w(P_S(uv))
\end{equation} 
The following property of greedy $(1+\epsilon)$-spanners is crucial in our analysis.  \ifFull\else The proof follows by contradiction to Equation~(\ref{eq:edge-path-ineq}) and can be found in Appendix~\ref{app:proofs}.
\fi
\begin{lemma} \label{lm:hereditary-prop}
Let $S$ be the greedy $(1+\epsilon)$-spanner of a graph and let $H$ be a subgraph of $S$.  Then the greedy $(1+\epsilon)$-spanner of $H$ is itself. 
\end{lemma}
\ifFull
\begin{proof}
Let $S_H$ be the greedy $(1+\epsilon)$-spanner of $H$.  Suppose for a contradiction that there is an edge $e=uv \in H\setminus S_H$. Since $uv$ is not added to the greedy spanner of $H$, there must be a $u$-to-$v$-path $P_{S_H}(uv)$ in $S_H$ that witnesses the fact that $uv$ is not added (i.e.\ $(1+\epsilon)w(e) > w(P_{S_H}(uv))$).
However, $P_{S_H}(uv) \subseteq S$, contradicting Equation~\ref{eq:edge-path-ineq}. 
\end{proof}
\fi

\subsection{Charging scheme}
To argue that $S$ is {\em light}, that is, has weight $L(\epsilon) w(\MST(G)$ for some function $L(\epsilon)$, we identify a specific {\em charging path} $P_S(uv)$ from $u$ to $v$ for each non-spanning-tree edge $uv$ of the spanner.  One may think of $P_S(uv)$ as being the shortest $u$-to-$v$ path in $S$ when $uv$ is added to the spanner, but this is not necessary for the analysis; we only need that $(1+\epsilon)w(uv) \leq w(P_S(uv))$ (as is guaranteed by Equation~(\ref{eq:edge-path-ineq}) for greedy spanners) for every path in $S$.  We call $(uv,P_S(uv))$ a \emph{charging pair}.  For a spanning tree $T$ (not necessarily a minimum spanning tree), we call a set of charging pairs $(e,P_S(e))$ for all edges $e \in S\backslash T$ a \emph{charging scheme}. We say that an edge is {\em charged to} if it belongs to the charging path for another edge. A charging scheme is \emph{acyclic} if for every edge $e \not \in T$, the directed graph where vertices of the graph are edges not in $T$ and directed edges represent charged to relationship, i.e. there is a directed edge $(e_1 \rightarrow e_0)$ if $e_0$ is charged to by $e_i$, is acyclic. \ifFull

\fi
A charging scheme is \emph{$k$-simple} if each edge $e \in S \backslash T$ is charged to at most once and each edge in $T$ is charged to at most $k$ times.  \ifFull If $S$ has a $k$-simple charging scheme then:
        \[
       (1 + \epsilon) w(S\backslash T) \leq \sum_{e \in S\backslash T} w(P_S(e))\leq k\cdot w(T) + w(S \backslash T)
        \]
where the first inequality follows from edges in $S\backslash T$ having charging paths and the second  inequality follows from each edge in $T$ appearing in charging paths at most $k$ times and each edge in $S\backslash T$ appearing in  charging paths at most once. Rearranging the left- and right-most sides of this inequality gives us:\else
Based on these definitions , one can prove (see Appendix~\ref{app:proofs} for full details):
\fi
\begin{lemma} \label{lem:charging}
If $S$ is a greedy $(1+\epsilon)$-spanner of a graph that has a $k$-simple acyclic charging scheme to a spanning tree $T$, then $w(S) \leq (1 + \frac{k}{\epsilon})w(T)$.
\end{lemma}
\noindent Indeed, in proving the greedy $(1+\epsilon)$-spanner has weight at most
$(1 + \frac{2}{\epsilon})\MST(G)$ when $G$ is planar, Althofer et al.~\cite{ADDJS93} implicitly proved the existence of a 2-simple acyclic charging scheme to $\MST(G)$. Since our paper only deals with acyclic simple charging schemes, we simply say \emph{simple charging schemes} to refer to acyclic simple charging schemes. We will use a stronger result for outer-planar graphs, planar graphs in which all the vertices are on the boundary of a common face (the outer face), which we take to be simple. \ifFull\else The proof of the following lemma is included in Appendix~\ref{app:proofs}\fi.
\begin{lemma} \label{lm:disk-charging}
If $G$ is an outer-planar graph and $T$ is a path formed by all the edges in the boundary less one edge, then $G$ has an acyclic $1$-simple charging scheme to $T$.
\end{lemma}
\ifFull \begin{proof}
Let $e$ be the edge on the boundary of $G$ that is not in $T$. Let $T^*$ be the spanning tree of the dual graph containing all the edges that do not correspond to edges of $T$. We construct a charging scheme for $G$ by traversing $T^*$ in post-order, considering all the non-outer faces.  Consider visiting face $f$ with children $f_1, \ldots, f_k$ and parent $f_0$.  Let $e_i$ be the edge of $G$ between $f$ and $f_i$ for all $i$ and let $P_0$ be the path between $e_0$'s endpoints in $T \cap \{e_1,e_2, \ldots, e_k\}$ that contains all the edges $\{e_1,e_2, \ldots, e_k\}$. Then, by Equation~\ref{eq:edge-path-ineq}, $(e_0,P_0)$ is a charging pair for $e_0$. We observe that the set of charging pairs produced from this process is a $1$-simple charging scheme to $T$.
\begin{figure}[h] 
  \centering
  \vspace{-20pt}
  \definecolor{ffqqqq}{rgb}{1,0,0}
\definecolor{uququq}{rgb}{0.25,0.25,0.25}
\definecolor{qqqqff}{rgb}{0,0,1}
\begin{tikzpicture}[line cap=round,line join=round,>=triangle 45,x=1.0cm,y=1.0cm]
\clip(0,1) rectangle (5,5.58);
\fill[line width=1.2pt,fill=black,fill opacity=0.1] (1.98,5) -- (1.1,4.38) -- (0.75,3.36) -- (1.07,2.33) -- (1.93,1.69) -- (3.01,1.67) -- (3.89,2.29) -- (4.24,3.31) -- (3.92,4.34) -- (3.06,4.98) -- cycle;
\draw [line width=1.2pt] (1.98,5)-- (1.1,4.38);
\draw [line width=1.2pt] (1.1,4.38)-- (0.75,3.36);
\draw [line width=1.2pt] (0.75,3.36)-- (1.07,2.33);
\draw [line width=1.2pt] (1.07,2.33)-- (1.93,1.69);
\draw [line width=1.2pt] (1.93,1.69)-- (3.01,1.67);
\draw [line width=1.2pt] (3.01,1.67)-- (3.89,2.29);
\draw [line width=1.2pt] (3.89,2.29)-- (4.24,3.31);
\draw [line width=1.2pt] (4.24,3.31)-- (3.92,4.34);
\draw [line width=1.2pt] (3.92,4.34)-- (3.06,4.98);
\draw [dash pattern=on 2pt off 2pt] (1.1,4.38)-- (1.93,1.69);
\draw [dash pattern=on 2pt off 2pt] (1.98,5)-- (1.93,1.69);
\draw [dash pattern=on 2pt off 2pt] (1.93,1.69)-- (3.89,2.29);
\draw [dash pattern=on 2pt off 2pt] (3.92,4.34)-- (3.89,2.29);
\draw [dash pattern=on 2pt off 2pt] (1.98,5)-- (3.92,4.34);
\draw [dash pattern=on 2pt off 2pt] (3.92,4.34)-- (1.93,1.69);
\draw [line width=1.2pt,dash pattern=on 2pt off 2pt,color=ffqqqq] (1.98,5)-- (3.06,4.98);
\draw (2.46,5.4) node[anchor=north west] {$e$};
\draw (1.84,5.44) node[anchor=north west] {$1$};
\draw (0.94,4.81) node[anchor=north west] {$2$};
\draw (0.45,3.75) node[anchor=north west] {$3$};
\draw (0.8,2.55) node[anchor=north west] {$4$};
\draw (1.73,1.8) node[anchor=north west] {$5$};
\draw (3.01,1.83) node[anchor=north west] {$6$};
\draw (3.89,2.44) node[anchor=north west] {$7$};
\draw (4.38,3.5) node[anchor=north west] {$8$};
\draw (4.05,4.62) node[anchor=north west] {$9$};
\draw [line width=1.2pt,dotted,color=qqqqff] (1.66,3.95)-- (1.16,3.01);
\draw [line width=1.2pt,dotted,color=qqqqff] (1.66,3.95)-- (2.77,3.73);
\draw [line width=1.2pt,dotted,color=qqqqff] (2.77,3.73)-- (3.32,2.66);
\draw [line width=1.2pt,dotted,color=qqqqff] (3.32,2.66)-- (2.98,1.85);
\draw [line width=1.2pt,dotted,color=qqqqff] (2.77,3.73)-- (3.06,4.79);
\draw [line width=1.2pt,dotted,color=qqqqff] (4.07,3.21)-- (3.32,2.66);
\draw (3.06,5.41) node[anchor=north west] {$10$};
\begin{scriptsize}
\fill [color=qqqqff] (1.98,5) circle (1.5pt);
\fill [color=qqqqff] (1.1,4.38) circle (1.5pt);
\fill [color=uququq] (0.75,3.36) circle (1.5pt);
\fill [color=uququq] (1.07,2.33) circle (1.5pt);
\fill [color=uququq] (1.93,1.69) circle (1.5pt);
\fill [color=uququq] (3.01,1.67) circle (1.5pt);
\fill [color=uququq] (3.89,2.29) circle (1.5pt);
\fill [color=uququq] (4.24,3.31) circle (1.5pt);
\fill [color=uququq] (3.92,4.34) circle (1.5pt);
\fill [color=uququq] (3.06,4.98) circle (1.5pt);
\fill [color=qqqqff] (1.16,3.01) circle (1.5pt);
\fill [color=qqqqff] (1.66,3.95) circle (1.5pt);
\fill [color=qqqqff] (2.77,3.73) circle (1.5pt);
\fill [color=qqqqff] (3.32,2.66) circle (1.5pt);
\fill [color=qqqqff] (2.98,1.85) circle (1.5pt);
\fill [color=qqqqff] (4.07,3.21) circle (1.5pt);
\fill [color=qqqqff] (3.06,4.79) circle (1.5pt);
\end{scriptsize}
\end{tikzpicture}
  \vspace{-15pt}
  \caption{ An outer planar graph $G$. Bold edges are edges of $T$, dashed edges are non-tree edges and dotted edges are edges of the dual spanning tree, less the dual of $e$.}
  \label{fig:disk-graph-charging}
\vspace{-10pt}
\end{figure}
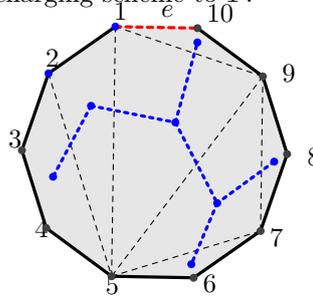
\end{proof}
\else

For a greedy spanner $S$, one may instead, when it is convenient, define a {\em weak} $k$-simple charging scheme for a supergraph of $S$. A weak $k$-simple charging scheme is a $k$-simple charging scheme in which Equation~(\ref{eq:edge-path-ineq}) need not hold for charging pairs.  \ifFull\else See Appendix~\ref{app:proofs} for the proof of the following. \fi

\begin{lemma}\label{lem:weak}
Let $S$ be a greedy spanner with spanning tree $T$ and let $\hat S$ be a supergraph of $S$ that $S$ spans.  If $\hat S$ has a weak $k$-simple charging scheme to $T$ then $S$ has a $k$-simple charging scheme to $T$.
\end{lemma}

\ifFull
\begin{proof}
  Consider an edge $\hat e \in \hat S \setminus S$. We first argue that $\hat S \setminus \{\hat e \}$ has a weak $k$-simple charging scheme. Since $T \subseteq S \cap \hat S$, $\hat e \notin T$ and so $\hat e$ can be charged to at most once. If $\hat e$ is not in the charging path of any edge in $S$, we simply ignore it. Otherwise, suppose that $\hat e$ is in the charging path $P_{\hat S}(e)$ for another edge $e$ of $\hat S$, then we define the charging path for $e$ to be the simple path between $e$'s endpoints that is in $P_{\hat S}(e) \cup P_{\hat S}(\hat e) \setminus  \{\hat e \}$.  The resulting set of paths is a weak $k$-simple charging scheme since every edge of $P_{\hat S}(\hat e)$ is charged to one fewer time (by the removal of $\hat e$) and at most once more (by $e$).

  By induction, $S$ has a weak $k$-simple charging scheme to $T$.  Since $S$ is a greedy spanner, Equation~(\ref{eq:edge-path-ineq}) holds for every charging pair, so the weak $k$-simple charging scheme to $T$ is a  $k$-simple charging scheme to $T$.
\end{proof}
\fi
\section{The greedy spanner of an \texorpdfstring{$H$}{H}-minor free graph is possibly light}\label{sec:h-minor-free-spanner}

Our result relies on the Graph Minor Structure Theorem due to Robertson and Seymour~\cite{RS03} which guarantees a structural decomposition of an $H$-minor-free graph into simpler graphs. \ifFull Note that we only rely on this decomposition for the analysis of the greedy spanner, not the construction, so we only need the existence of the decomposition (although algorithms also exist). \fi Informally, the seminal Graph Minor Structure Theorem of Robertson and Seymour states that every $H$-minor-free graph is the (small) clique-sum of graphs that are {\em almost} embeddable on graphs of small genus. We give a formal statement of the Graph Minor Structure Theorem below after some requisite definitions.


We first argue that almost-embeddable graphs have light $(1+\epsilon)$-spanners assuming bounded treewidth graphs have light $(1+\epsilon)$-spanners. We partition the spanner edges of almost-embeddable graphs into two parts: those in the surface-embeddable part and those in the non-embeddable part. We bound the weight of the surface-embeddable part by ``cutting along'' a subset of edges to create an outer-planar graph and then using the lightness bound for outer-planar graphs. Since the large-grid minor of the graph must be contained in the surface-embeddable part, we can show that the non-embeddable part has bounded treewidth. Therefore, the lightness of $(1+\epsilon)$-spanners of the non-embeddable part follows from the assumption that bounded treewidth graphs have light $(1+\epsilon)$-spanners.

\subsection{Definitions: Treewidth, pathwidth and the structure of \texorpdfstring{$H$}{H}-minor free graphs}

Note that if $G$ excludes $K_{|H|}$ as a minor, then it also excludes $H$.  
 
 \subparagraph*{Tree decomposition} \emph{A tree decomposition} of $G(V,E)$ is a pair $(\mathcal{X},\mathcal{T})$ where $\mathcal{X} = \{X_i| i \in I\}$, each $X_i$ is a subset of $V$ (called bags), $I$ is the set of indices,  and $\mathcal{T}$ is a tree whose set of nodes is $\mathcal{X}$ satisfying the following conditions:
\begin{enumerate} [noitemsep,nolistsep]
  \item The union of all sets $X_i$ is $V$.
  \item For each edge $uv \in E$, there is a bag $X_i$ containing both $u,v$.
  \item For a vertex $v \in V$, all the bags containing $v$ make up a subtree of $\mathcal{T}$.  
  \end{enumerate}
The \emph{width} of a tree decomposition $\mathcal{T}$ is $\max_{ i \in I }|X_i| -1$ and the treewidth of $G$, denoted by $\tw$, is the minimum width among all possible tree decompositions of $G$.  A \emph{path decomposition} of a graph $G(V,E)$ is a tree decomposition where the underlying tree is a path and is denoted by $(\mathcal{X},\mathcal{P})$. The pathwidth of a $G(V,E)$, denoted by $\pw$, is defined similarly. 

 \subparagraph*{$\beta$-almost-embeddable} A graph $G$ is \emph{$\beta$-almost-embeddable} if there is a set of vertices $A \subseteq V(G)$ and $\beta+1$ graphs $G_0,G_1, \ldots, G_\beta$ such that: 
	\begin{enumerate}[nolistsep,noitemsep]
	\item $|A| \leq \beta$.
	\item  $G_0 \cup G_1 \cup \ldots \cup G_\beta = G[V\backslash A]$.
	\item $G_0$ is embeddable in a surface $\Sigma$ of genus at most $\beta$.
	\item Each $G_j$ has a path decomposition $(\mathcal{X}_j,\mathcal{P}_j)$ of width at most $\beta$ and length $|I_j|$, which is the number of bags, for $j \geq 1$.
	\item There are $\beta$ faces $F_1,F_2, \ldots, F_\beta$ of $G_0$ such that $|F_j| \geq |I_j|$, $G_0 \cap G_j  \subseteq V(F_j)$, and the vertices of $F_j$ appear in the bags of $\mathcal{X}_j$ in order along $\mathcal{P}_j$ for each $j$.
	\end{enumerate}
The vertices $A$ are called \emph{apices} and the graphs $\{G_1, G_2, \ldots, G_\beta\}$ are called \emph{vortices}. The vortex $G_j$ is said to be \emph{attached} to the face $F_j$, $1 \leq j \leq \beta$. 

 \subparagraph*{$\beta$-clique-sum} Given two graphs $H_1,H_2$, a graph $H$ is called a {\em $\beta$-clique-sum} of $H_1$ and $H_2$ if it can be obtained by identifying a clique  of size at most $\beta$ in each of two graphs $H_1,H_2$ and deleting some of the clique edges. 
 
\noindent We can now state Robertson and Seymour's result:

\subparagraph*{Graph Minor Structure Theorem} (Theorem~1.3~\cite{RS03}).
{\em An $H$-minor-free graph can be decomposed into a set of $\beta(|H|)$-almost-embeddable graphs that are glued together in a tree-like structure by taking $\beta(|H|)$-clique-sums where $\beta(|H|)$ is a function of $|H|$.}

\bigskip
\noindent It will  be convenient to consider a simplified decomposition that assumes there are no edges between the apices and vortices.  This simplification introduces zero-weight edges that do not change the distance metric of the graph. \ifFull\else We include the proof of this claim in Appendix~\ref{app:proofs}.\fi

\begin{claim}\label{clm:embeddable-const}
  There is a representation of a $\beta$-almost-embeddable graph as a $2\beta$-almost-embeddable graph that has no edge between apices and vertices of vortices (that are not in the surface-embedded part of the graph) and that maintains the distance metric of the graph.
\end{claim}
\ifFull
\begin{proof}
Let $V_a$ be the set of vertices in vortex $V$ that are adjacent to apex $a$.  Split vertex $a$ into two vertices $a$ and $a_V$ connected by a zero-weight edge so that $a_V$'s neighbors are $V_a \cup \{a\}$ and so that contracting the zero-weight edge gives the original graph.  Add $a_V$ to all of the bags of the path decomposition of $V$.  Now all the edges that connected $a$ to $V$ are within the vortex.  

  Consider the face in the surface-embedded part of the $\beta$-almost-embeddable graph to which $V$ is attached and let $xy$ be edge in that face that is between the first and last bags of $V$ and such that $y$ is in the first bag of $V$.  Add the edges $xa_V$ and $a_Vy$ to the embedded part of the graph and give them weight equal to the distance between their endpoints.  Now $xa_V$ is the edge in that face that is between the first and last bags of the vortex and $a$ is adjacent to a vertex that is in the surface-embedded part of the $\beta$-almost-embeddable graph.  See Figure~\ref{fig:apex-reduction}.

  The {\em splitting} of $a$ into $a_v$ increases the pathwidth of the vortex by 1.  Repeating this process for all apex-vortex pairs increases the pathwidth of each vortex by at most $\beta$.
\begin{figure}[h] 
  \centering
  \vspace{-20pt}
  \input{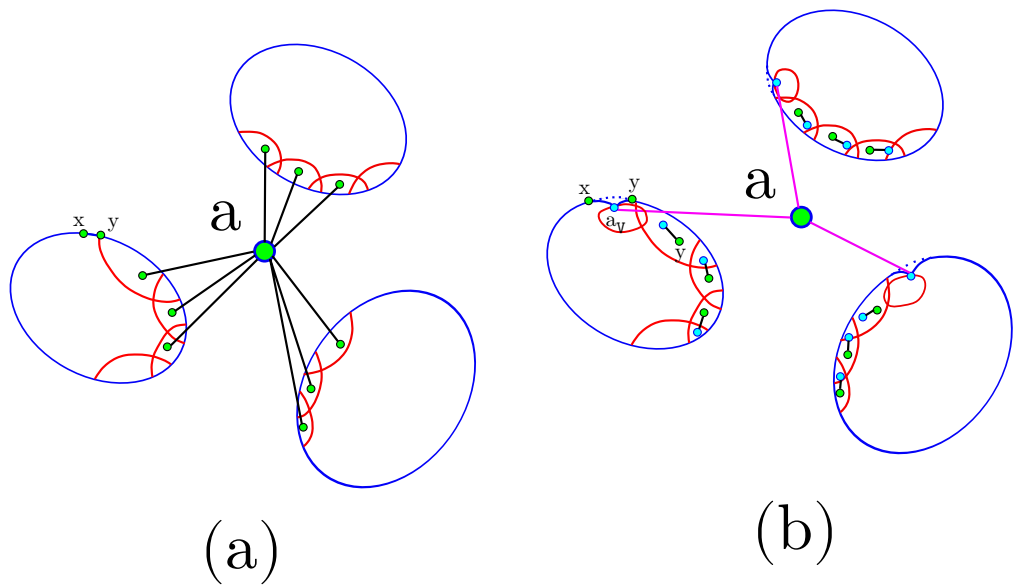}
  \vspace{-15pt}
  \caption{Apex $a$ and the vortices it is adjacent to before (a) and after (b) the reduction of Claim~\ref{clm:embeddable-const}}
  \label{fig:apex-reduction}
\vspace{-10pt}
\end{figure}
\end{proof}
\fi

In the remainder of this section, we prove Theorem~\ref{thm:h-minor-spanner} by assuming that, if $S$ is the greedy $(1+\epsilon)$-spanner of a graph of treewidth $\tw$:
\begin{equation}
w(S) \le g(\tw,\epsilon)\cdot w(\MST)\label{eq:tw-span}
\end{equation}
The bulk of the technical detail in dealing with $H$-minor free graphs is in handling the vortices, which we do first.
\subsection{Handling vortices}

  In this subsection, we consider a greedy $(1+\epsilon)$-spanner $G$ of some apex-free $\beta$-almost-embeddable graph.  We will show that:
\begin{equation}
  \label{eq:weightJ}
  w(G) \le q(\beta,\epsilon)\cdot w(\MST(G))
\end{equation}
where $q(\beta,\epsilon)$ is the lightness that only depends on $\beta$ and $\epsilon$. Note that the MST of the graph is the MST of the spanner. We assume that $G$ is connected since we can bound the weight of each component separately. Let $G = G_0 \cup G_1 \cup \ldots \cup G_\beta$ be the decomposition of $G$ into a graph $G_0$ embedded on a surface $\Sigma$ of genus at most $\beta$ and a set $G_1, \ldots, G_\beta$ vortices according to the definition of $\beta$-almost-embeddability.  Let $C_i$ be the cycle bounding the face of $G_0$ to which vortex $G_i$ is attached.

\vspace{1mm} \noindent {\bf Bounding the weight of the vortices}
First we bound the weight of the vortices and their bounding cycles $(\bigcup_{i = 1}^\beta C_i) \cup G_i$ by showing that $\MST(G) \cup (\bigcup_{i = 1}^\beta C_i \cup G_i)$ has bounded treewidth.

Let $K = (\MST(G) \cap G_0) \cup \bigcup_{i = 1}^\beta C_i$ and let $K^*$ be the dual of $K$; the vertices of $K^*$ correspond to the faces of $K$.  Consider a vertex $v$ of $K^*$ that does not correspond to a face bounded by a cycle in $\{C_1,C_2,\ldots, C_\beta\}$.  Then $v$ must be adjacent to a face that is bounded by a cycle $C_j$ for some $j$ because $K \setminus (\bigcup_{i = 1}^\beta C_i)$ is a forest.  Therefore, the diameter of $K^*$ is $O(\beta)$.
Since a graph of genus $\beta$ has treewidth $O(\beta \cdot \text{diameter})$ (Eppstein, Theorem 2~\cite{Eppstein00}), $K^*$ has treewidth $O(\beta^2)$.  Since the dual of a graph of treewidth $\tw$ and genus $\beta$ has treewidth $O(\tw+\beta)$ (Mazoit, Proposition 2~\cite{Mazoit12}), $K$ has treewidth $O(\beta^2)$. 

Grohe showed that if $G_i$ is a vortex attached to a face of $K$, then  $\tw(G_i \cup K) \leq (\pw(G_i) + 1)(\tw(K) + 1)-1$ (Lemma 2~\cite{Grohe03}). Adding in each of the vortices $G_1, G_2, \ldots, G_\beta$ to $K$  and using Grohe's result gives that the treewidth of $\MST(G) \cup (\bigcup_{i = 1}^\beta C_i \cup G_i)$ is $O(\beta^{\beta+2})$. 
Since $\MST(G) \cup (\bigcup_{i = 1}^\beta C_i \cup G_i)$ is the subgraph of a greedy spanner, $\MST(G) \cup (\bigcup_{i = 1}^\beta C_i \cup G_i)$ is a greedy spanner itself (Lemma~\ref{lm:hereditary-prop}) and by Equation~(\ref{eq:tw-span}), we have:
\begin{equation}
  \label{eq:vortices}
  \textstyle\sum_{i=1}^\beta \left(w(C_i) + w(G_i)\right) \le g(O(\beta^{\beta+2}),\epsilon)\cdot w(\MST(G)).
\end{equation}

\vspace{1mm} \noindent {\bf Bounding the surface-embedded part of the spanner}
Let $\widehat{G}$ be the graph obtained from $G_0$ by contracting the cycles  $\{C_1,C_2, \ldots, C_\beta\}$ bounding the vortices into vertices $\{c_1,c_2, \ldots,c_\beta\}$, removing loops and removing parallel edges. Let $T_{\widehat{G}}$ be the minimum spanning tree of $\widehat{G}$ and let $X$ be the set of edges that has smallest summed weight such that cutting open the surface $\Sigma$ along $T_{\widehat{G}} \cup X$ creates a disk; $|X| \le 2\beta$ (see, e.g.~Eppstein~\cite{Eppstein03}).  Since $T_{\widehat{G}} \subseteq \MST(G)$ and an edge in the spanner is also the shortest path between its endpoints, we have:
	\begin{equation} \label{eq:light-Hx}
	w(T_{\widehat{G}} \cup X) \leq (2\beta+1)w(\MST(G))
	\end{equation}
Cutting the surface open along $X \cup T_{\widehat{G}} \cup \left(\bigcup_{i=1}^\beta C_i \right)$ creates $\beta+1$ disks: one disk for each face that a vortex is attached to and one disk $\Delta$ corresponding to the remainder of the surface.  The boundary of $\Delta$, $\partial \Delta$, is formed by two copies of each of the edges of $T_{\widehat{G}} \cup X$ and one copy of each of the edges in $\cup_{i=1}^\beta C_i$ (see Figure~\ref{fig:cutting-surface} in  Appendix~\ref{app:figures}). Therefore we can use Equations~(\ref{eq:vortices}) and~(\ref{eq:light-Hx}) to bound the weight of the boundary of $\Delta$:
\begin{equation}
  \label{eq:delta}
  w(\partial \Delta) \leq (4 \beta + 2 + g(O(\beta^{\beta+2}),\epsilon))\cdot w(\MST(G))
\end{equation}
\ifFull
\begin{figure}[hb] 
  \centering
  \vspace{-20pt}
  \input{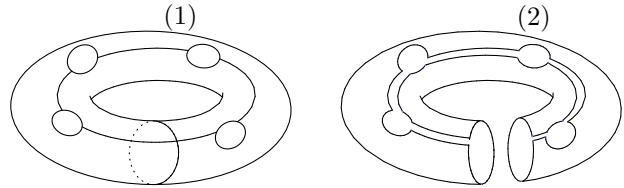}
  \vspace{-15pt}
  \caption{The surface $(2)$ is obtained by cutting the surface $\Sigma$ in $(1)$ along $T_{\widehat{G}} \cup X$. Small ovals are cycles $C_1,C_2,\ldots, C_\beta$.}
  \label{fig:cutting-surface}
\end{figure}
\fi
Let $E_\Delta$ be the set of edges of $G$ that are not in $\MST(G)$, $X$ or any of the vortices or their boundaries.  That is, $E_\Delta$ contains all the edges that we have not yet bounded.  Since $\partial \Delta$ spans $G_0$, there is a 1-simple charging scheme to $\partial \Delta$ (less an edge, Lemma~\ref{lm:disk-charging}).  Therefore, by Lemma~\ref{lem:charging} and Lemma~\ref{lm:hereditary-prop},
\begin{equation}
  \label{eq:rest}
  w(E_\Delta) \le (1 + 1/\epsilon)\cdot (4 \beta + 2 + g(O(\beta^{\beta+2}),\epsilon))\cdot w(\MST(G))
\end{equation}

\vspace{1mm} \noindent {\bf Total weight of the spanner} Since every edge of $G$ is either in $\MST(G)$, $X$, $E_\Delta$ or $G_i \cup C_i$ (for some $i$), summing Equations~(\ref{eq:vortices}),~(\ref{eq:light-Hx}) and~(\ref{eq:rest}) gives us Equation~(\ref{eq:weightJ}):
\begin{equation}
	\begin{split}
	w(G) &\le \underbrace{\left(2 + 2\beta + (1+1/\epsilon)(4 \beta + 2 + g(O(\beta^{\beta+2}),\epsilon)) + g(O(\beta^{\beta+2}), \epsilon) \right)}_{q(\beta, \epsilon)}w(\MST(G))\\
	\end{split}
\end{equation}

\subsection{Adding apices and clique-sums}

We are now ready to 
prove Theorem~\ref{thm:h-minor-spanner} by considering the apices and clique-sums of the decomposition.
Let $G = J_1\oplus J_2\oplus \ldots \oplus J_\gamma$ be the $\beta(|H|)$-clique-sum of $\beta(|H|)$-almost-embeddable graphs $J_1, J_2, \ldots, J_\gamma$ given by the Graph Minor Structure Theorem. For $J_i$, let $A_i$ be its set of apices, let $J_i^1, \ldots, J_i^\beta$ be its set of vortices and let $J_i^0$ be the graph embedded on a surface of genus at most $\beta$, as provided by the definition of $\beta(|H|)$-almost-embeddable.  We assume the representation includes no edges between apices and the internal vertices of vortices (vertices that are not in $J_i^0$) by Claim~\ref{clm:embeddable-const}.

Let $S$ be the greedy $(1+\epsilon)$-spanner of $G$. For each $J_i$, we define $S_i$ as the set of spanner edges in the apex-free $\beta(|H|)$-almost-embeddable part of $J_i$ (formally, $S_i = S \cap (\cup_\ell J_i^\ell)$).  Consider the spanning forest of $S_i$ that is induced by $\MST(G)$: $F_i = \MST(G) \cap S_i$.  We choose a subset of edges $E_i$ of $S_i \backslash F_i$ such that:
	\begin{enumerate}[nolistsep,noitemsep]
	\item[(i)] The number of components of $F_i \cup E_i$ is minimized.
	\item[(ii)] Subject to (i), the size of $E_i$ is minimized.
	\item[(iii)] Subject to (i) and (ii), the weight of $E_i$ is minimized.
     \end{enumerate}	 
By the choice of $E_i$ and since $A_i$ has no edges to the internal vertices of vortices, each tree of $F_i \cup E_i$ is a minimum spanning tree for each apex-free $\beta(|H|)$-almost-embeddable component of $S_i$.  Since $S_i$ is a subgraph of a greedy spanner, it is its own greedy spanner (Lemma~\ref{lm:hereditary-prop}) and so, by Equation~(\ref{eq:weightJ}), we have $w(S_i) \leq q(\beta(|H|), \epsilon)(w(F_i) + w(E_i))$. Summing over $i$, we have:	
	\begin{equation}\label{eq:H-i-weight}
	\begin{split}
		\textstyle\sum_{i=1}^\gamma w(S_i) & \leq q(\beta(|H|), \epsilon)\textstyle\sum_{i=1}^\gamma \Big(w(F_i) + w(E_i)\Big)\\  &\leq q(\beta(|H|), \epsilon)\Big(w(\MST(G)) + \textstyle\sum_{i=1}^\gamma w(E_i)\Big)
	\end{split}
	\end{equation}
Let $S(A_i)$ be edges of $S$ incident to vertices in $A_i$. Then, $S \cap J_i = S_i \cup S(A_i)$ and hence $S = \bigcup_{i=1}^\gamma (S_i \cup S(A_i))$. Therefore, we have:
	\begin{equation}\label{eq:spanner-weight}
	w(S) \leq \textstyle\sum_{i=1}^\gamma w(S_i) + \textstyle\sum_{i=1}^\gamma w(S(A_i))
	\end{equation}
Now define $J_i' = F_i \cup E_i \cup S(A_i)$. Then $\tw(J_i') \leq |A_i| + 1 \leq \beta(|H|) + 1$. Let $G' = J'_1 \oplus J'_2 \oplus \ldots \oplus J'_\gamma$.  We get that $\tw(G') \leq \max_i \tw(J'_i)  \leq \beta(|H|) + 1$ by a result of Demaine et al.~(Lemma 3~\cite{DHNRT04}). Note that $\MST(G') = \MST(G)$. We have (by Lemma~\ref{lm:hereditary-prop} and Equation~(\ref{eq:tw-span})): 
	\begin{equation} \label{eq: clique-sum-weight}
	 \textstyle\sum_{i=1}^\gamma w(F_i) + w (E_i) + w(S(A_i)) \leq g(\beta(|H|) + 1,\epsilon) \cdot w(\MST(G))
	\end{equation}
By Equations~\ref{eq:H-i-weight},~\ref{eq:spanner-weight} and~\ref{eq: clique-sum-weight}, we get  Theorem~\ref{thm:h-minor-spanner}:
	\begin{equation*}
	w(S) \leq \Big( (q(\beta(|H|), \epsilon) + 1)\cdot g(\beta(|H|) + 1,\epsilon)  + q(\beta(|H|), \epsilon)\Big)w(\MST(G) ).
	\end{equation*}	 

\section{The greedy spanner for bounded pathwidth graphs is light}\label{sc:bounded-pw}

Grigni and Hung proved that graphs of pathwidth $\pw$ have a $(1+\epsilon)$-spanner of lightness $O(\pw^3/\epsilon)$~\cite{GH12}. They do so by building a spanning tree that is {\em monotone} with respect to the path decomposition of weight $O(\pw^2)\,w(\MST)$ (Lemma 2~\cite{GH12}) and devising what we observe to be  an $O(\pw)$-simple charging scheme to the monotone spanning tree (Lemma 3~\cite{GH12}). \ifFull Unfortunately, monotonicity is highly dependent on the graph decomposition being a path and not a tree and these techniques are unlikely to generalize to tree decompositions. Grigni and Hung~\cite{GH12} constructed an example in which the monotone spanning tree of a bounded treewidth graph has weight $\Omega(\log(n))w(\MST)$. \fi We prove that graphs of pathwidth $\pw$ have light greedy $(1+\epsilon)$-spanners by showing that there is an $O(\pw^2)$-simple charging scheme to the $\MST$, forgoing the need for constructing a monotone spanning tree, giving Theorem~\ref{thm:tw-spanner}. Our proof gives an evidence that one can avoid the pathwidth-specific monotonicity argument, opening a door to show that graphs of bounded treewidth may have light greedy $(1+\epsilon)$-spanners as well.  We discuss the challenges for bounded treewidth graphs at the end of the paper. Throughout this section $G$ refers to the greedy $(1+\epsilon)$-spanner of some graph of pathwidth $\pw$.



\subparagraph*{Smooth decompositions} It will be convenient for our proofs to work with a standardized path decomposition. We assume that bags are ordered linearly, i.e, $X_i$ and $X_{i+1}$ are adjacent ($1 \leq i \leq |I|-1$), and the path decomposition is {\em smooth}. A path decomposition $(\mathcal{X},\mathcal{P})$ is \emph{smooth} if $|X_i| = \pw+1 \quad \forall i \in I$ and $|X_i \cap X_{i+1}| = \pw $  for all  $ 1 \leq i \leq |I|-1$. Bodlaender~\cite{Bodlaender93B} showed that a path decomposition can be turned into a smooth path decomposition of the same width in linear time. We root the path decomposition $(\mathcal{X},\mathcal{P})$ at the bag $X_{|I|}$. 
 
For adjacent bags $X_i,X_{i+1}$, we call the vertex in $X_{i+1} \backslash X_{i}$ the \emph{introduced} vertex of $X_{i+1}$ and the vertex in $X_{i} \backslash  X_{i+1}$ the \emph{forgotten} vertex of $X_i$. All vertices of $X_1$ are introduced vertices and all vertices of $X_{|I|}$ are forgotten vertices. 

\subparagraph*{Overview: designing an $O(\pw^2)$-simple charging scheme} In designing  an $O(\pw^2)$-simple charging scheme, one needs to guarantee (i) each non-tree edge is charged at most once and (ii) each tree edge\footnote{A tree edge is an edge of the minimum spanning tree.} is charged at most $O(\pw^2)$ times. At high level, we use the charging scheme for edges in $X_0 \cup \ldots \cup X_{i-1}$ to design a charging scheme for the edges introduced to $X_i$.  Let $u$ be the introduce vertex of $X_i$.  We need to define charging pairs for all non-tree edges between vertices of $u$ and $X_i \setminus \{u\}$.

The simpler case is when there is a tree edge $uv$ incident to $u$ in $X_i$.  For a non-tree edge $wu$ incident to $u$ in $X_i$, we define a charging path for $wu$ using the edges $uv$ (a tree edge) and $vw$ (an edge that already has a defined charging path since $vw$ is in a descendant bag of $X_i$).  (This will be formalized as the {\em triangle rule}.)  However, to guarantee condition (i), we must prevent the use of $wu$ in charging paths in the future.  We keep track of this by way of a {\em charging forest} whose vertices are the edges of $G$;  in this case we add an edge to the charging forest connecting $uv$ and $wu$.



The harder case is when there is no tree edge incident to $u$ in $X_i$.  In this case, we consider the $u$-to-$v$ spanning-tree path that contains only edges of ancestor bags of $X_i$.  We use this path as a sit-in for the edge $uv$ of the previous case.  To guarantee condition (ii), we must be careful to not use tree-edges in ancestor bags too many times.  Since $u$ may have an ancestral spanning-tree path to multiple vertices of $X_i \setminus \{u\}$, we delay the choice of which paths to use in defining a charging pair for edges incident to $u$ by adding {\em dashed edges} to the charging forest corresponding to all possible constructions.  Then, to achieve condition (ii), we carefully select which dashed edges to convert in defining the charging pairs for edges incident to $u$ in $X_i$.

\subparagraph*{Normalized graph} We simplify the presentation of the formal argument, we use a {\em normalized graph} which merges a graph $G$ with its smooth path decomposition $(\mathcal{X},\mathcal{P})$. For each bag $X_i$, define the \emph{bag graph} $G_i = (X_i, E_i)$ to be a subgraph of $G$ where $E_i$ is a maximal subset of edges of $G[X_i]$ incident to introduced vertices of $X_i$. This implies that each edge of $\MST(G)$ appears in exactly one bag graph. 
 For adjacent bags $X_i,X_{i+1}$, we add edges between two copies of the same vertex of $G$ in $G_i$ and $G_{i+1}$. We call the resulting graph the \emph{normalized graph} of $G$ with respect to the path decomposition $(\mathcal{X},\mathcal{P})$ and denote it by $G_\mathcal{P}(V_\mathcal{P},E_\mathcal{P})$. We assign weight 0 to edges between bag graphs and weight $w(e)$ for the copy of the edge $e$ in $G$. See Figure~\ref{fig:normalized-graph} in Appendix~\ref{app:figures} for an example.
\ifFull For an example of a normalized graph,  
\begin{figure}[h] 
  \centering
  \vspace{-20pt}
  \input{../figs/normalized-graph.tex}
  \vspace{-15pt}
  \caption{The normalized graph (3) of a graph (1) with path decomposition (2).} \label{fig:normalized-graph}
\vspace{-10pt}
\end{figure}
\else
\fi
Since the distances between vertices in $G$ and the distances between their copies in $G_\mathcal{P}$ are the same:
\ifFull		\begin{equation} \label{eq:mst-preserving}
		w(\MST(G_\mathcal{P})) = w(\MST(G))
		\end{equation}
\else $w(\MST(G_\mathcal{P})) = w(\MST(G))$. \fi
Further, since distances are preserved, we can define a charging scheme for the greedy $(1+\epsilon)$-spanner of $G_\mathcal{P}$ to $\MST(G_\mathcal{P})$.  In fact, we prove:

\begin{theorem} \label{thm:simple-charging-mst}
There is a $O(\pw^2)$-simple charging scheme for the greedy $(1+\epsilon)$-spanner of $G_\mathcal{P}$ to $\MST(G_\mathcal{P})$.
\end{theorem}

\noindent This theorem, along with the equivalence of the metrics of $G_\mathcal{P}$ and $G$ and Lemma~\ref{lem:charging}, gives Theorem~\ref{thm:tw-spanner}.  To simplify the construction of the charging scheme, we assume that $G$ is a $k$-path; this is without loss of generality by Lemma~\ref{lem:weak}.

\subsection{The charging forest}

The main difficulty in defining the charging scheme is the existence of introduced vertices that are connected to the $\MST$ via an edge that is in an ancestor bag of the path decomposition. For the other types of introduced vertices, there is a triangle in the vertex's bag graph that allows us to pay for the non-tree edges incident to that vertex.  Throughout this section $\MST$ refers to $\MST(G_\mathcal{P})$.

To define the $O(\pw^2)$-simple charging scheme, we construct a charging forest $\Phi$ to guide the charging. The charging forest is a rooted spanning forest of the line graph\footnote{The nodes of the line graph of a graph $G$ are the edges of $G$; the edges of the line graph are between nodes whose corresponding edges of $G$ share an endpoint.} of $G$, with one tree rooted at a vertex corresponding to each edge of $\MST(G)$ (that is, each non-zero edge of the $\MST(G_P)$).   We call the nodes of $\Phi$ {\em $\phi$-vertices} and denote a $\phi$-vertex of $\Phi$ by $(u,v)$ where $uv$ is an edge of $G$.

We use three types of edges in constructing $\Phi$: \emph{dashed edges}, \emph{bold edges} and \emph{mixed edges}. 
We construct $\Phi$ iteratively; $\Phi_i$ will be an intermediate forest of the line graph of $G[\cup_{j \le i} X_i]$.  $\Phi_i$ may contain all three types of edges but $\Phi = \Phi_{|I|}$ will contain only bold and mixed edges.  From $\Phi_{i-1}$ to $\Phi_i$, a dashed edge may be deleted or converted to a mixed edge and newly added edges will either all be bold or all be dashed.  A {\em dashed-free tree} of $\Phi_i$ is a maximal tree of $\Phi_i$ that contains no dashed edge.  Trees of the intermediate forests may be unrooted, but each tree of an intermediate forest will contain at most one root.

We also maintain a \emph{contracted forest} $\Lambda_i$ spanning vertices of $X_i$. Intuitively, $\Lambda_i$ tells us the connections between vertices of $X_i$ in $\MST[\cup_{j \leq i} X_j]$. $\Lambda_i$ is used to handle introduced vertices that have no tree edges to other vertices of the same bag. We also assign unique positive rank to each edge of $\MST$ in $G$ and assign rank 0 to all edges in $\MST$ added to $G_\mathcal{P}$ by the normalizing process. Ranks of edges of $\MST$ are used to define edge-rank in $\Lambda_i$ as follows: the rank of an edge $uv$ in $\Lambda_i$, denoted by $r_i(uv),$ is the minimum rank over the edges in the $u$-to-$v$ path of $\MST[\cup_{j \leq i}X_i]$. We assign rank to each edge of $\MST$ in a way that the rank of each edge in $\Lambda_i$ is unique. 

\subparagraph*{The triangle rule} We say that an edge $((u,v),(u,w))$ of the line graph satisfies the triangle rule if  $vw \in \MST(G)$ and $(u,v)$ and $(u,w)$ are in distinct dashed-free trees, both of which do not contain roots (i.e.\ $\phi$-vertices that correspond to edges of $\MST(G)$).  We will add edges that satisfy the triangle rule to the charging forest.  To maintain the acyclicity of the intermediate charging forests, if adding $((u,v),(u,w))$ introduces a cycle, the most recently added dashed edge on the path in the charging forest from $(u,v)$ to $(u,w)$ is deleted.

\subparagraph*{Invariants of the charging forest} Let $\Phi_i$ be the charging forest for $G_\mathcal{P}[\cup_{j \le i} X_j]$. We say that a $\phi$-vertex $(u,v)$ of $\Phi_i$ is \emph{active} if $u,v \in X_{i+1}$. We will show that $\Phi_i$ satisfies the following invariants:
\begin{enumerate}[nolistsep,noitemsep]
\item[(i)] For two trees $T_1$ and $T_2$ of $\MST[\cup_{j \le i} X_j]$, all the $\phi$-vertices of the form $(u,v)$ such that $u \in T_1$ and $v \in T_2$ are in a common unrooted tree of $\Phi_i$. Further every unrooted tree of $\Phi_i$ contains $\phi$-vertices spanning components of $\MST[\cup_{j \leq i} X_j]$.
\item[(ii)] For a tree $T$ of $\MST[\cup_{j \le i} X_j]$, all the $\phi$-vertices of the form $(u,v)$ such that $u,v \in T$ are in rooted dashed-free trees of $\Phi_i$.
\item[(iii)] Each unrooted dashed-free tree of $\Phi_i$ contains at least one active $\phi$-vertex.
\item[(iv)] For any $j \leq i$ and any two distinct trees $T_1, T_2$ of $\MST[\cup_{\ell = j}^i X_\ell]$, $(y_{j,1},y_{j,2})$ and $(y_{i,1},y_{i,2})$ are in the same unrooted dashed-free tree of $\Phi_i$ where $y_{\ell,k} \in T_k \cap X_\ell$ for $\ell = i,j$ and $k = 1,2$.
\end{enumerate}

\subsubsection{Initializing the charging forest}

We define $\Phi_1$ as a forest of only bold edges.  Recall that $G_1$ is a complete graph so there is a $\phi$-vertex for every unordered pair of vertices in $X_1$.  We greedily include bold edges in $\Phi_1$ that satisfy the triangle rule.  Equivalently, consider the subgraph $H$ of the line graph of $G_1$ consisting of edges $((u,v),(u,w))$ where $vw \in \MST$;  $\Phi_1$ corresponds to any maximal forest of $H$ each tree of which contains at most one root ($\phi$-vertex corresponding to an edge of $\MST$).  

We arbitrarily assign an unique rank to each edge of $\MST[X_1]$ from the set $\{1,2, \ldots, m_1\}$ where $m_1$ is the number of edges of $\MST[X_1]$ so that the highest-ranked edge is incident to the forgotten vertex of $X_1$.  We define $\Lambda_1$ to be the forest obtained from $\MST[X_1]$ by contracting the highest-ranked edge incident to the forgotten vertex of $\MST[X_1]$. We observe that each edge in $\Lambda_1$ has a unique rank, since the rank of each edge of $\MST[X_1]$ is unique. \ifFull \else We prove that $\Phi_1$ satisfies four invariants in Appendix~\ref{app:proofs}.
\begin{claim}\label{clm:init-inv}
Charging forest $\Phi_1$ satisfies all invariants.
\end{claim}
\fi
\subsubsection{Growing the charging forest}

We build $\Phi_i$ from $\Phi_{i-1}$ and show that $\Phi_i$ will satisfy the invariants using the fact that $\Phi_{i-1}$ satisfies the invariants.  Let $u$ be the introduced vertex of non-leaf bag $X_i$. If $u$ isolated in $\MST[X_i]$, we say that $u$ is a free vertex.  The construction of $\Phi_i$ from $\Phi_{i-1}$ depends on whether or not $u$ is free.
\subsubsection*{$u$ is a free vertex} 

Let $\Lambda_{i-1}$ be the contracted forest of $X_{i-1}$. To obtain $\Phi_i$ from $\Phi_{i-1}$, we add $\phi$-vertices $(u,w) \quad \forall w \in (X_i \backslash \{u\})$ and add dashed edges $((u,w_j),(u,w_k))$ for each edge $w_jw_k$ in $\Lambda_{i-1}$. We assign rank to the dashed edge $((u,w_j),(u,w_k))$ to be the rank of $w_jw_k$ in $\Lambda_{i-1}$. We additionally change some dashed edges to mixed edges which we describe below. Since converting dashed edges to mixed edges does not affect these invariants, we have: \ifFull
First, we show that $\Phi_i$ satisfies Invariants~(i),~(ii) and ~(iv), since converting dashed edges to mixed edges does not affect these invariants:
\begin{enumerate}[nolistsep,noitemsep]
\item[(i)] The only new tree of $\MST[\cup_{j \le i} X_\ell]$ compared to $\MST[\cup_{j \le i-1} X_\ell]$ is $u$.  For any pair of trees in $\MST[\cup_{\ell = j}^{i-1} X_\ell]$, Invariant~(i) holds for $\Phi_i$ because it helps for $\Phi_{i-1}$.  For a component $C$ of $\Lambda_{i-1}$, the addition of the edges $((u,w_j),(u,w_k))$ creates a new (unrooted) component spanning all $\phi$-vertices $(u,v)$ where $v$ is in the corresponding component of $\MST[\cup_{j \le i} X_j]$.  Therefore, Invariant~(i) holds for $\Phi_i$.
\item[(ii)] Since $u$ is free and no new $\phi$-vertices of the form $(u,v)$ where $u$ and $v$ are in the same tree of $\MST[\cup_{j \le i} X_j]$ are introduced, Invariant~(ii) holds for $\Phi_i$ because it holds for $\Phi_{i-1}$. 
\item[(iv)] As with $(i)$, the only new tree of $\MST[\cup_{\ell = j}^i X_\ell]$ is $u$, which only has a non-zero intersection with $X_i$.  So, for $j < i$, Invariant~(iv) holds for $\Phi_i$ because it holds for $\Phi_{i-1}$.  For $j = i$, all the components of $\MST[\cup_{\ell = j}^i X_\ell]$ are isolated vertices, so the invariant holds trivially.
\end{enumerate}
\else
\begin{claim}\label{cl:free-inv}
Charging forest $\Phi_i$ satisfies Invariants~(i),~(ii) and ~(iv).
\end{claim}
We include the formal proof of Claim~\ref{cl:free-inv} in Appendix~\ref{app:proofs}.
\fi

\subparagraph*{Converting dashed edges to mixed edges} To ensure that $\Phi_i$ will satisfy Invariant~(iii), we convert some dashed edges to mixed edges.  First, consider the newly-added $\phi$-vertex $(u,v)$ where $v$ is the forgotten vertex of $X_i$.  To ensure that $(u,v)$ will be in a dashed-free tree that contains an active $\phi$-vertex, we convert the added dashed edge $((u,v),(u,w))$, that has highest rank among newly added dashed edges, to a mixed edge.  Second, suppose that $\tau$ is a dashed-free tree of $\Phi_{i-1}$ such that its active $\phi$-vertex is of the form $(v,w)$ for some $w \in X_{i}\backslash \{u\}$.  The tree $\tau$ is at risk of becoming inactive in $\Phi_i$.  However, by Invariant~(i), $\tau$ is a subtree of a component $\hat \tau$ of $\Phi_{i-1}$ that consists of $\phi$-vertices $(u_1,u_2)$ such that $u_k \in T_k$ where $T_k$ is a component of $\MST[\cup_{j \le i-1} X_i]$ for $k = 1,2$.
Further, since $T_1$ and $T_2$ must be connected in $\MST$ by a path through ancestor nodes of the path decomposition, there must be vertices in $T_k \cap X_{i+1}$; w.l.o.g., we take $(u_1,u_2)$ to be an active $\phi$-vertex in $\hat \tau$.  We greedily convert dashed edges in $\hat \tau$ to mixed edges to grow $\tau$ until such an active $\phi$-vertex is connected to $(v,w)$ by mixed and bold edges, thus ensuring that $\Phi_i$ will satisfy Invariant~(iii).  We break ties between dashed edges for conversion to mixed edges by selecting the dashed edges that were first added to the charging forest and if there are multiple dashed edges which were added to the charging forest at the same time, we break ties by converting the dashed edges that have higher ranks to mixed edges.  Note that by converting dashed edges into mixed edges, we add edges between dashed-free trees. Thus, Invariant~(i),~(ii),~(iv) are unchanged.

Finally, we need to update $\Lambda_i$ from $\Lambda_{i-1}$. Let $\{v_1,v_2, \ldots, v_q\}$ be the set of neighbors of $v$ such that $r_{i-1}(vv_1) \leq \ldots \leq r_{i-1}(vv_q)$. Delete $v$ from $\Lambda_{i-1}$ and if $q \geq 2$, add edges $v_jv_{j+1}$ ($1 \leq j \leq q-1$) to $\Lambda_{i-1}$ to obtained the forest $\Lambda_i$. Then, by definition, $r_i(v_{j}v_{j+1}) = r_{i-1}(vv_j)$ and $r_i(xy) = r_{i-1}(xy)$ for all other edges of $\Lambda_i$. Thus, by induction hypothesis, the rank of each edge of $\Lambda_{i}$ is unique.

\subsubsection*{$u$ is not a free vertex}

To build $\Phi_i$ from $\Phi_{i-1}$, we add $\phi$-vertices $(u,w) \forall w \in X_i \backslash \{u\}$ and greedily add bold edges that satisfy the triangle rule. We include the formal proof that $\Phi_i$ satisfies Invariant~(i),~(ii), and ~(iv) in Appendix~\ref{app:proofs}.
 \ifFull In the case that $u$ has no tree edge to the forgotten vertex of $X_i$, we will additionally change some dashed edges to mixed edges which we describe below; since this only affects Invariant~(i), we describe these conversions when we discuss Invariant~(iii).  We prove that $\Phi_i$ satisfies each of the invariants in turn.

\vspace{1mm} \noindent {\bf Invariant (i)} Let $T_1$ and $T_2$ be distinct trees of $\MST[\cup_{j \leq i} X_j]$.  If $T_1$ and $T_2$ are distinct trees of $\MST[\cup_{j \leq i-1} X_j]$, then the invariant holds for $\Phi_i$ because it holds for $\Phi_{i-1}$.  Otherwise, we may assume w.l.o.g.\ that $T_1$ contains $u$ and a subtree $T_3$ that is a component of $\MST[\cup_{j \leq i-1} X_j]$.  Let $x$ be a vertex of $T_3 \cap X_i$ and $w$ be a vertex of $T_2 \cap X_i$.  Then $ux$ is an edge of $\MST$ and $(u,w),(w,x)$ will be connected in $\Phi_i$ by greedy applications of the triangle rule.  By the same argument as used for showing that $\Phi_1$ satisfies Invariant~(i), $(u,w)$ will not be connected to a root $\phi$-vertex (a $\phi$-vertex corresponding to a tree edge) in $\Phi_i$.	

\vspace{1mm} \noindent {\bf Invariant (ii)} We need only prove this for the tree $T$ of $\MST[\cup_{j \leq i} X_i]$ that contains $u$ as other cases are covered by the fact that  $\Phi_{i-1}$ satisfies Invariant~(ii).  Let $T_1$ and $T_2$ be distinct trees of $\MST[\cup_{j \leq i-1} X_i]$ that are subtrees of $T$.  Let $u_1$ and $u_2$ be vertices of $T_1 \cap X_i$ and $T_2 \cap X_i$, respectively.  

We start by showing that $(u_1,u_2)$ is in a rooted dashed free tree of $\Phi_i$.  Let $v_1$ and $v_2$ be the neighbors of $u$ in $T_1$ and $T_2$, respectively.  (Note, it may be that, e.g., $u_1 = v_1$.)  By Invariant~(ii), $(v_2,u_2)$ is in a rooted dashed-free tree of $\Phi_{i-1}$.  By the triangle rule, $(v_1,u_2)$ will be connected by (a possibly non-trivial sequence of) bold edges to $(u,u_2)$ which in turn will be connected by (a possibly non-trivial sequence of) bold edges to $(v_2,u_2)$.  Therefore, $(v_1,u_2)$ will be in a rooted dashed-free tree of $\Phi_{i}$. In this next paragraph, we show that $(u_1,u_2)$  and $(v_1,u_2)$ are in the same dashed-free tree of $\Phi_{i-1}$, which we have just shown belongs to a rooted dashed-free tree of $\Phi_{i}$, showing that $(u_1,u_2)$ is in a rooted dashed free tree of $\Phi_i$.

To show that $(u_1,u_2)$  and $(v_1,u_2)$ are in the same dashed-free tree of $\Phi_{i-1}$, consider an index $k$ such that $u_1$ and $v_1$ are connected to $x_1$ in  $\MST[\cup_{\ell = k}^i X_\ell]$ and $u_2$ is connected to $x_2$ in $\MST[\cup_{\ell = k}^i X_\ell]$.  Then by Invariant~(iv) for $\Phi_{i-1}$, $(x_1,x_2)$ is in the same dashed-free tree as $(u_1,u_2)$ and $(v_1,u_2)$. 

Now consider a $\phi$-vertex $(x_1,x_2)$ where $x_k \in T_k \cap X_j$  for $k = 1,2$ and $j < i$.  Let $T_k'$ be the subtree of $T_k$ that is in $\MST[\cup_{\ell = j}^{i-1} X_\ell]$ and let $u_k = T_k' \cap X_i$ for $k = 1,2$.  We just showed that $(u_1,u_2)$ is in a rooted dashed-free tree of $\Phi_{i}$.  By Invariant~(iv) for $\Phi_{i-1}$, $(x_1,x_2)$ and $(u_1,u_2)$ are in the same dashed-free tree.  Therefore, $(x_1,x_2)$ is in a rooted dashed-free tree of $\Phi_{i}$. 

\vspace{1mm} \noindent {\bf Invariant (iii)}  Let $v$ be the forgotten vertex of $X_i$.

If $u = v$, then any $\phi$-vertex that was active in $\Phi_{i-1}$ is still active in $\Phi_i$. Thus, $(u,w)$ is connected to existing $\phi$-vertices that remain active for any $w \in X_i \setminus \{u\}$. 

If $u \ne v$ and $uv \in \MST[X_i]$, then $\phi$-vertex $(x,v)$ will become inactive in $\Phi_i$.  However, $(x,v)$ is connected to $(x,u)$ by the triangle rule or it was already connected to a root in a dashed-free tree. In either case, $\Phi_i$ satisfy Invariant~(iii).

Otherwise, to ensure that previously active $\phi$-vertices of the form $(v,w)$ for $w \in X_{i}$ get connected to dashed-free trees that contain active $\phi$-vertices, we convert dashed edges to mixed edges.  We do so in the same way as for the above described method for when $u$ was a free vertex, guaranteeing that $\Phi_i$ will satisfy Invariant~(iii).

\vspace{1mm} \noindent {\bf Invariant (iv)} Consider $j \leq i$, trees $T_1, T_2$ of $\MST[\cup_{\ell = j}^i X_\ell]$, and $\phi$-vertices $(y_{j,1},y_{j,2})$ and $(y_{i,1},y_{i,2})$  as defined in Invariant~(iv).  For the case $j = i$, the proof that $\Phi_{i}$ satisfies Invariant~(iv) is the same as for $\Phi_1$.  Further, if $u \notin T_1,T_2$, $\Phi_{i}$ satisfies Invariant~(iv) because $\Phi_{i-1}$ satisfies Invariant~(iv), therefore, we assume w.l.o.g.\ that $u \in T_1$.  Finally, consider the components of $T_1$ in $\MST[\cup_{\ell = j}^{i-1} X_\ell]$; if $y_{j,1}$ and $y_{i,1}$ are in the same component, then $\Phi_{i}$ satisfies Invariant~(iv) because $\Phi_{i-1}$ satisfies Invariant~(iv).  Therefore, we assume they are in different components, $T_1^p$ and $T_1^q$, respectively.  

Let $x$ be the neighbor of $u$ in $T_1^p$.  By Invariant~(iv) for $\Phi_{i-1}$, $(x,y_{i,2})$ and $(y_{j,1},y_{j,2})$ are in a common unrooted dashed-free tree.  We show that $(x,y_{i,2})$ and $(y_{i,1},y_{i,2})$ are in a common unrooted dashed-free tree of $\Phi_i$, proving this invariant is held.  Let $y$ be the neighbor of $u$ in $T_1^q$.
\begin{enumerate}[nolistsep,noitemsep]
\item By the triangle rule, $(x,y_{i,2})$ will get connected by (a possibly non-trivial sequence of) bold edges to $(u,y_{i,2})$ because $xu \in \MST$ and $(u,y_{i,2})$ will get connected by (a possibly non-trivial sequence of) bold edges to $(y, y_{i,2})$ because $yu \in \MST$.
\item Let $k$ be the index such that $y$ and $y_{i,1}$ are connected to $a$ in $\MST[\cup_{\ell = k}^i X_\ell]$ and $y_{i,2}$ is connected to $b$ in $\MST[\cup_{\ell = k}^i X_\ell]$.  Then by Invariant~(iv), $(a,b)$ is in the same unrooted dashed-free tree as $(y, y_{i,2})$ and $(y_{i,1},y_{i,2})$.
\end{enumerate}
Together these connections show that $(x,y_{i,2})$ and $(y_{i,1},y_{i,2})$ are in a common unrooted dashed-free tree.

\else
\begin{claim}\label{cl:non-free-inv}
Charging forest $\Phi_i$ satisfies Invariants~(i),~(ii) and ~(iv).
\end{claim}

We show that $\Phi_i$ satisfies Invariant~(iii) here. Let $v$ be the forgotten vertex of $X_i$. Note that $u$ may have no tree edge to the forgotten vertex of $X_i$. In this case, we will ensure that previously active $\phi$-vertices of the form $(v,w)$ for $w \in X_{i}$ get connected to dashed-free trees that contain active $\phi$-vertices, we convert dashed edges to mixed edges.  We do so in the same way as for the above described method for when $u$ was a free vertex, guaranteeing that $\Phi_i$ will satisfy Invariant~(iii). Otherwise, we consider two subcases:

\begin{enumerate}[noitemsep,nolistsep]
\item If $u = v$, then any $\phi$-vertex that was active in $\Phi_{i-1}$ is still active in $\Phi_i$. Thus, $(u,w)$ is connected to existing $\phi$-vertices that remain active for any $w \in X_i \setminus \{u\}$. 
\item If $u \ne v$, then $\phi$-vertex $(x,v)$ will become inactive in $\Phi_i$.  However, $(x,v)$ is connected to $(x,u)$ by the triangle rule or it was already connected to a root in a dashed-free tree. In either case, $\Phi_i$ satisfy Invariant~(iii).
\end{enumerate}

\fi 
\noindent Finally, we show how to update $\Lambda_i$. Let $v$ be the forgotten vertex of $X_i$ and $u_1,u_2, \ldots, u_p$ be neighbors of $u$ in $\MST[X_i]$ such that $u_1 = v$ if $uv \in \MST[X_i]$. Let $r'$ be the maximum rank of over ranked edges of  $\MST[\cup_{j \leq i-1}X_j]$. We assign rank $r_i(u,u_j) = r' + j$ for all $1 \leq j \leq p$. We will update $\Lambda_i$ from $\Lambda_{i-1}$ depending on the relationship between $u$ and $v$: 

\begin{enumerate}[nolistsep,noitemsep]
\item If $u = v$, then we add edges $p-1$ edges $u_ju_{j+1}$, $1 \leq j \leq p-1$, to  $\Lambda_{i-1}$. 
\item If $u\not= v$ and $uv \in \MST[X_i]$, we replace $v$ in $\Lambda_{i-1}$ by $u$ and add $p$ edges $uu_j$ to $\Lambda_{i-1}$. 
\item Otherwise, we add $u$ and $p$ edges $uu_j$ to $\Lambda_{i-1}$. Let $\{v_1,v_2, \ldots,v_q\}$ be the set of neighbors of $v$ such that $r_{i-1}(vv_1) \leq \ldots \leq r_{i-1}(vv_q)$. We delete $v$ from $\Lambda_{i-1}$ and add $q-1$ edges $v_jv_{j+1}$. 
\end{enumerate}

\begin{claim}\label{cl:lamda-dist-rank}
Each edge of the forest $\Lambda_i$ has a distinct rank.
\end{claim}
\ifFull See Figure~\ref{fig:charging-mst} for an example of a charging forest and contracted forests
\begin{figure}[p] 
  \centering
  \vspace{-20pt}
  \input{../figs/charging-mst.tex}\\
  \input{../figs/shorten-tree-ex.tex}
  \input{../figs/charging-graph.tex}
  \vspace{-15pt}
  \caption{ Top: A normalized graph and its MST. Bold edges are edges of the MST. Non-tree edges of the first bag graphs are dotted. Non-tree edges of other bag graphs are incident to the introduced vertex of these bags and not shown in the figure. Center: Contracted forest $\Lambda_1, \ldots, \Lambda_{12}$. Circled numbers are ranks of the edges of the contracted forests. Bottom: Charging forest for the graph. Dotted edges are mixed edges. Circled vertices are roots of the charging forest}
  \label{fig:charging-mst}
\vspace{-10pt}
\end{figure}
\else 
We include the proof of Claim~\ref{cl:lamda-dist-rank} in Appendix~\ref{app:proofs}. See Figure~\ref{fig:charging-mst} in Appendix~\ref{app:figures} for an example of a charging forest and contracted forests.
\fi
    
\subsubsection{Using the charging forest to define an \texorpdfstring{$O(\pw^2)$}{pw}-simple charging scheme}

By Invariant~(ii), all trees in $\Phi$ are rooted at $\phi$-vertices correspond to edges of $\MST$. Order the edges of $E(G)\backslash \MST(G)$ given by DFS pre-order of $\Phi$. Let $(u_{i-1},v_{i-1})$ and $(u_i,v_i)$ ($i \geq 2$) be two $\phi$-vertices of a component $T$ of $\Phi$ in this order. Define $P_i$ to be the (unique) $u_i$-to-$v_i$ path in $T\cup \{(u_{i-1},v_{i-1})\}$ that contains the edge  $(u_{i-1},v_{i-1})$.  We take $(P_i,(u_i,v_i))$ to be the charging pair for $(u_i,v_i)$. Note that the roots of $\Phi$ are edges of $\MST(G)$, so the charging paths are well-defined.

 We prove that the charging scheme defined by the charging pairs is $O(\pw^2)$-simple by bounding the number of times non-zero edges of $\MST$ are charged to. By the triangle rule,  if  $((u,v),(u,w))$ is a bold edge of $\Phi$, then $vw \in \MST$. We call the set of 3 vertices $\{u,v,w\}$ a \emph{charging triangle}. If $((u,v),(u,w))$  is a mixed edge, $vw \notin \MST(G)$. In this case, we call $\{u,v,w\}$ a \emph{charging pseudo-triangle}. The $v$-to-$w$ path in $\MST(G)$ is called the \emph{pseudo-edge}; $vw$ is said to be \emph{associated} with the charging (pseudo-) triangle. We say the edge $((u,v),(u,w))$ \emph{represents} the charging (pseudo-) triangle. 

\begin{claim} \label{clm:associated-triangle}
There are at most $\pw-2$ charging triangles associated with each non-zero edge of $\MST$.
\end{claim}
\begin{proof}
A charging triangle consists of one non-zero weight edge of $\MST$ and one edge not in the $\MST$ in the same bag graph. Note that each non-zero weight edge of $\MST$ is in exactly one bag graph and  each bag graph has at most $\pw-2$ edges not in the $\MST$; that implies the claim. 
\end{proof}
\ifFull \else \noindent The proof of the following is in Appendix~\ref{app:proofs}.\fi
\begin{lemma} \label{lm:associated-pseudo-triangle}
Each edge of $\MST(G)$ is in the paths corresponding to the pseudo-edges of at most $2\pw^2$ charging pseudo-triangles.  
\end{lemma}

\ifFull We investigate how dashed edges are changed into mixed edges as this is when a charging pseudo-triangle arises. Recall that dashed edges are added to $\Phi$ when we process free vertices. Let $u$ be a free vertex that is introduced in bag $X_i$ and let $((u,v_1),(u,v_2))$ be a mixed edge (that is, a dashed edge in $\Phi_i$ that is later converted to a mixed edge).  Then by our construction, $v_1$ and $v_2$ are in the same component of $\MST[\cup_{k \le i-1}X_k]$. For each $j \geq  i$ and $w \in X_i$, let $T^{ij}_w$ be the component of $ \MST[\cup_{i \leq k \leq j}X_k]$ containing $w$. We will say that an introduced vertex of a bag $X_j$ is {\em branching} if it is not free, is not the forgotten vertex of $X_j$ and is not connected to the forgotten vertex of $X_j$ by $\MST$.  Note that dashed edges are converted to mixed edges when processing free and branching vertices.
  We say that a tree $T^{ij}_w \in  \MST[\cup_{i \leq k \leq j}X_k]$ is \emph{forgotten} in $X_j$ if $j = |I|$ or :
	\begin{itemize}[nolistsep,noitemsep]
	\item $T_w^{ij} \cap X_j = \{v\}$ which is the forgotten vertex of $X_j$.
	\item the introduced vertex of $X_j$ is branching or free.
	\end{itemize}
We note that at most one tree can be forgotten in each bag $X_j$. 

For each vertex $x \in \Lambda_{i-1}$, let $N_{i-1}^r(x)$ be the set of vertices in $\Lambda_{i-1}$ reachable from $x$ via paths consisting of edges of ranks larger than $r$ in $\Lambda_{i-1}$. Let $F_{x,r}^{ij} = \cup_{ w\in N_{i-1}^r(x) \cup \{x\}}T^{ij}_{w}$. We say that the forest $F_{x,r}^{ij}$ is forgotten in $X_j$ if every tree in $F_{x,r}^{ij}$ is forgotten in $X_{j'}$ for some $j ' \leq j$ and $F_{x,r}^{ij} \cap X_j \not= \emptyset$.

\begin{lemma} \label{lm:realized-dashed-edges}
Let $r$ be the rank of the edge $v_1v_2 \in \Lambda_{i-1}$. If the dashed edge $((u,v_1),(u,v_2))$ is changed into a mixed edge, there exists a bag $X_j$ for $j \geq i$ such that $T^{ij}_u \cap \{F^{ij}_{v_1,r} \cup F^{ij}_{v_2,r}\} = \emptyset$, $T^{ij}_{u} \cap X_j \not= \emptyset$ and exactly one of the forests  $F_{v_1,r}^{ij}, F_{v_2,r}^{ij}$ is forgotten in $X_j$.
\end{lemma}

We now bound the number of pseudo-triangles that contain an edge $e$ of $\MST(G)$.  Let $X_i$ be the bag that containing $e$ and $\Lambda_i$ be the corresponding contracted forest.  Let $u_0,v_0$ be two vertices of $\Lambda_i$ in the same tree of $\Lambda_i$ such that $e$ is in the $u_0$-to-$v_0$ path $P_\MST(u_0,v_0)$. We say that a pseudo-triangle \emph{strongly contains} $P_\MST(u_0,v_0)$ if $P_\MST(u_0,v_0)$ is a subpath of the pseudo-edge and the rank of every edge of the psedudo-edge of the triangle is at least the minimum rank over edges of $P_\MST(u_0,v_0)$. 
 
\begin{lemma}\label{lm:num-nesting}
There are at most $\pw$ pseudo-triangles strongly containing $P_\MST(u_0,v_0)$. 
\end{lemma} 

\begin{proof}
Let $\Delta_1, \Delta_2, \ldots, \Delta_q $  be pseudo-triangles that strongly contain $P_\MST(u_0,v_0)$. Let $\Delta_k = \{u_k,v_k,w_k\}$ and $X_{s_k}$ be the bag that has $w_k$ as the introduced vertex ($1 \leq k \leq q$). Let $r$ be the minimum rank over edges of $P_\MST(u_0,v_0)$ and $X_\ell$ be the bag in which one of the forests $F^{i\ell}_{u_0,r}, F^{i\ell}_{v_0,r}$, say $F^{i\ell}_{u_0,r}$, is forgotten. Then, $u_k \in F^{i\ell}_{u_0,r}$ and $v_k \in F^{i\ell,r}_{v_0}$. We can assume w.l.o.g\ that $s_1 \leq s_2 \leq \ldots \leq s_q$.  Then, $F_{u_k,r}^{s_k\ell} = F^{i\ell}_{u_0,r} \cap \MST[\cup_{s_k \leq j \leq \ell}X_j]$ and $F_{v_k,r}^{s_k\ell} = F^{i\ell}_{v_0,r} \cap \MST[\cup_{s_k \leq j \leq \ell}X_j]$. By Lemma~\ref{lm:realized-dashed-edges}, $T^{s_k\ell}_{w_k} \cap X_{\ell} \not= \emptyset$. Furthermore, all the tree $T^{s_k\ell}_{w_k}$ must be disjoint since otherwise, say $T^{s_j\ell}_{w_j}$ is a subtree of $T^{s_k\ell}_{w_k}$ $( k \leq j)$,  the second case of the proof of Lemma~\ref{lm:realized-dashed-edges} implies that the dashed edge $((w_j,u_j),(w_j,v_j))$ is removed from $\Phi$. Hence, $q \leq \pw$.
\end{proof}

\begin{lemma} \label{lm:associated-pseudo-triangles}
Each edge of $\MST(G)$ is in the paths corresponding to the pseudo-edges of at most $2\pw^2$ charging pseudo-triangles.  
\end{lemma}
\begin{proof}
Let $e$ be an arbitrary edge of $\MST(G)$. Let $X_i$ be the bag that containing $e$ and $\Lambda_i$ be the corresponding contracted forest. By the way we build the contracted forest, there are at most two edges, say $u_0v_0$ and $v_0w_0$, of $\Lambda_i$ that are incident to the same vertex $v_0$ such that $e \in P_\MST(u_0,v_0) \cap P_\MST(v_0,w_0)$. We observe that any pseudo-triangle that contains $e$ in the speudo-edge must contain the path $P_\MST(u,v)$ for some $u,v \in \Lambda_i$ such that one of two edges $u_0v_0$, $v_0w_0$ is in the path $P_{\Lambda_{i-1}}(u,v)$. Let $\Gamma(u_0,v_0)$ be the set of pseudo-triangles that have $P_{\Lambda_{i-1}}(u,v)$ containing $u_0v_0$. We define $\Gamma(v_0,w_0)$ similarly. We will show that $|\Gamma(u_0,v_0)|,|\Gamma(u_0,w_0)| \leq \pw^2$, thereby, proving the lemma. 

We only need to show that $|\Gamma(u_0,v_0)| \leq \pw^2$ since $|\Gamma(u_0,w_0)| \leq \pw^2$ can be proved similarly.  Let $\Delta_1, \Delta_2, \ldots, \Delta_p$ be the pseudo-triangles containing $e$ in the pseudo-forest such that for each $k$:
	\begin{enumerate}[nolistsep,noitemsep]
	\item Each triangle $\Delta_k$ contains distinct path $P_\MST(u_k,v_k)$ as subpath in the pseudo-edge. 
	\item Edge $u_0v_0$ is in the path $P_{\Lambda_{i-1}}(u_k,v_k)$ of $\Lambda_{i-1}$.
	\item Each path $P_\MST(u_k,v_k)$ has distinct rank.
	\end{enumerate}	 
Then, by the way we construct the contracted forest, if the minimum rank over edges of $P_\MST(u_j,v_j)$ is smaller than the minimum rank over edges of $P_\MST(u_k,v_k)$, then $P_{\Lambda_i}(u_k,v_k) \subsetneq P_{\Lambda_i}(u_j,v_j)$ for $1 \leq j \not= k \leq p$.  Therefore, we can rearrange $\Delta_1, \Delta_2, \ldots, \Delta_k$ such that $P_{\Lambda_{i}}(u_1,v_1)\subsetneq P_{\Lambda_{i}}(u_2,v_2) \subsetneq \ldots \subsetneq P_{\Lambda_{i}}(u_p,v_p)$. Thus, $p \leq \pw$ . By Lemma~\ref{lm:num-nesting}, for each $k$, there are at most $\pw$ pseudo-triangles containing the same subpath $P_\MST(u_k,v_k)$. Therefore, $|\Gamma(u_0,v_0)| \leq \pw^2$. 
\end{proof}
\else \fi

Consider each charging pair $(P_i,(u_i,v_i))$  in which $P_i$ contains $(u_{i-1},v_{i-1})$ that precedes $(u_i,v_i)$ in the DFS pre-order of a given component of $\Phi$. Let $(u_{i},v_{i}) = (\hat{u}_0,\hat{v}_0), \ldots, (\hat{u}_t,\hat{v}_t) = (u_{i-1},v_{i-1})$ be the set of $\phi$-vertices of the path between $(u_{i},v_{i})$ and $(u_{i-1},v_{i-1})$ in $\Phi$. We define $Q_j$ ($1 \leq j \leq t$) be the $\hat{u}_{j-1}$-to-$ \hat{v}_{j-1}$ path of $\MST\cup \{(\hat{u}_j,\hat{v}_j)\}$ containing edge $\hat{u}_j\hat{v}_j$. Then we have:
			\[
			P_i = Q_1 \ominus Q_2 \ominus \ldots \ominus Q_{t}
			\]
where $\ominus$ is the symmetric difference between two sets. Hence, charging tree edges of $P_i$ is equivalent to charging the tree edges of $Q_1, Q_2, \ldots, Q_{t}$. We observe that tree edges of $Q_j$ are the tree edges of the (pseudo-) triangle represented by $((\hat{u}_j,\hat{v}_j),(\hat{u}_{j-1},\hat{v}_{j-1}))$. Since each edge of $\Phi$ appears twice in the collection of paths between  $(u_{i},v_{i})$ and $(u_{i-1},v_{i-1})$ in $\Phi$ for all $i$ and since different edges of $\Phi$ represents different charging (pseudo-) triangles, the tree edges of each charging (pseudo-) triangle are charged to twice. By Claim~\ref{clm:associated-triangle} and Lemma~\ref{lm:associated-pseudo-triangle}, each non-zero tree edge is charged to by at most $2(\pw-2) + 2(2\pw^2)= O(\pw^2)$ times. Hence, the charging scheme is $O(\pw^2)$-simple.
\subsection{Toward light spanners for bounded treewidth graphs}

The main difficulty in designing a simple charging scheme for bounded pathwidth graphs is the existence of free vertices. We introduce dashed edges to the charging forest when we handle free vertices and change a subset of dashed edges into mixed edges. By changing a dashed edge, say $((w,u),(w,v))$ into a mixed edge, we charge to the $u$-to-$v$ path in the MST once. Unfortunately, for bounded treewidth graphs, such charging can be very expensive. However, we observe that the $\phi$-vertex $(v,w)$ is contained in a rooted tree of $\Phi$. That means we can used $vw$ to charge to one of two edges $uv$ or $uw$. In general, for each non-tree edge $uv$ of the spanner in a bag $X_i$, we say $uv$ is a \emph{simple edge} if $u,v$ are in the same component of the MST restricted to descendant bags of $X_i$ only or ancestor bags of $X_i$ only. Simple edges can be paid for in the spanner ``cheaply''. Now, if $uv$ be an edge of the spanner in $X_i$ such that the $u$-to-$v$ path $P_\MST(u,v)$ of $\MST$ crosses back and forth through $X_i$. Let $w_1,w_2, \ldots, w_k$ be the set of vertices of $X_i$ in this order on the path $P_\MST(u,v)$. Then, we can charge the edge $uv$ by the set of simple edges $uw_1,w_1w_2,\ldots, w_kv$. We believe that this idea, with further refinement, will prove the existence of light spanners for bounded treewdith graphs.
\pagebreak





\bibliography{hung}
\ifFull
\else

\newpage
\appendix

\section{Notation}\label{app:notation}

We denote $G =(V(G),E(G))$ to be the graph with vertex set $V(G)$ and edge set $E(G)$ and use $n,m$ to denote the number of vertices and edges, respectively. The order of $G$, denoted by $|G|$, is the number of vertices of $G$. Each edge $e$ of $E(G)$ is a assigned a weight $w(e)$. We define $w(H) = \sum_{e \in E(H)}w(e)$ to be the weight of edges of a subgraph $H$ of $G$. The minimum spanning tree of $G$ is denoted by $\MST(G)$. For two vertices $u,v$, we denote the shortest distance  between them by $d_G(u,v)$. Given a subset of vertices $S$ and $u$ a vertex of $G$, we define $d_G(u,S) = \min_{v \in S}\{d_G(u,v)\}$ for $v = \arg \min_{v \in S}\{d_G(u,v)\}$. We omit the subscript $G$ when $G$ is clear from context. The subgraph of $G$ induced by $S$ is denoted by $G[S]$.

\section{Omitted Proofs}\label{app:proofs}

\begin{proof}[Proof of Lemma~\ref{lm:hereditary-prop}]
Let $S_H$ be the greedy $(1+\epsilon)$-spanner of $H$. We will prove that $S_H = H$. Note that $S_H$ is a subgraph of $H$  and shares the same set of vertices with $H$. Suppose for a contradiction that there is an edge $e=uv \in H\setminus S_H$. Since $uv$ is not added to the greedy spanner of $H$, there must be a $u$-to-$v$-path $P_{S_H}(uv)$ in $S_H$ that witnesses the fact that $uv$ is not added (i.e.\ $(1+\epsilon)w(e) > w(P_{S_H}(uv))$). However, $P_{S_H}(uv) \subseteq S$, contradicting Equation~\ref{eq:edge-path-ineq}. 
\end{proof}

\begin{proof}[Proof of Lemma~\ref{lem:charging}]
If $S$ has a $k$-simple charging scheme then:
        \[
       (1 + \epsilon) w(S\backslash T) \leq \sum_{e \in S\backslash T} w(P_S(e))\leq k\cdot w(T) + w(S \backslash T)
        \]
where the first inequality follows from edges in $S\backslash T$ having charging paths and the second  inequality follows from each edge in $T$ appearing in charging paths at most $k$ times and each edge in $S\backslash T$ appearing in  charging paths at most once. Rearranging the left- and right-most sides of this inequality gives us Lemma~\ref{lem:charging}.
\end{proof}

\begin{proof}[Proof of Lemma~\ref{lm:disk-charging}]
Let $e$ be the edge on the boundary of $G$ that is not in $T$. Let $T^*$ be the spanning tree of the dual graph containing all the edges that do not correspond to edges of $T$. We construct a charging scheme for $G$ by traversing $T^*$ in post-order, considering all the non-outer faces.  Consider visiting face $f$ with children $f_1, \ldots, f_k$ and parent $f_0$.  Let $e_i$ be the edge of $G$ between $f$ and $f_i$ for all $i$ and let $P_0$ be the path between $e_0$'s endpoints in $T \cup \{e_1,e_2, \ldots, e_k\}$ that contains all the edges $\{e_1,e_2, \ldots, e_k\}$. Then, by Equation~\ref{eq:edge-path-ineq}, $(e_0,P_0)$ is a charging pair for $e_0$. Also, since we visit $T^*$ in post-order, none of the edges $\{e_1,e_2, \ldots, e_k\}$ will be charged to when we build charging pairs for higher-ordered edges which are edges between faces of higher orders. Thus, the set of charging pairs produced from this process is a $1$-simple charging scheme to $T$.
\begin{figure}[h] 
  \centering
  \vspace{-20pt}
  \definecolor{ffqqqq}{rgb}{1,0,0}
\definecolor{uququq}{rgb}{0.25,0.25,0.25}
\definecolor{qqqqff}{rgb}{0,0,1}
\begin{tikzpicture}[line cap=round,line join=round,>=triangle 45,x=1.0cm,y=1.0cm]
\clip(0,1) rectangle (5,5.58);
\fill[line width=1.2pt,fill=black,fill opacity=0.1] (1.98,5) -- (1.1,4.38) -- (0.75,3.36) -- (1.07,2.33) -- (1.93,1.69) -- (3.01,1.67) -- (3.89,2.29) -- (4.24,3.31) -- (3.92,4.34) -- (3.06,4.98) -- cycle;
\draw [line width=1.2pt] (1.98,5)-- (1.1,4.38);
\draw [line width=1.2pt] (1.1,4.38)-- (0.75,3.36);
\draw [line width=1.2pt] (0.75,3.36)-- (1.07,2.33);
\draw [line width=1.2pt] (1.07,2.33)-- (1.93,1.69);
\draw [line width=1.2pt] (1.93,1.69)-- (3.01,1.67);
\draw [line width=1.2pt] (3.01,1.67)-- (3.89,2.29);
\draw [line width=1.2pt] (3.89,2.29)-- (4.24,3.31);
\draw [line width=1.2pt] (4.24,3.31)-- (3.92,4.34);
\draw [line width=1.2pt] (3.92,4.34)-- (3.06,4.98);
\draw [dash pattern=on 2pt off 2pt] (1.1,4.38)-- (1.93,1.69);
\draw [dash pattern=on 2pt off 2pt] (1.98,5)-- (1.93,1.69);
\draw [dash pattern=on 2pt off 2pt] (1.93,1.69)-- (3.89,2.29);
\draw [dash pattern=on 2pt off 2pt] (3.92,4.34)-- (3.89,2.29);
\draw [dash pattern=on 2pt off 2pt] (1.98,5)-- (3.92,4.34);
\draw [dash pattern=on 2pt off 2pt] (3.92,4.34)-- (1.93,1.69);
\draw [line width=1.2pt,dash pattern=on 2pt off 2pt,color=ffqqqq] (1.98,5)-- (3.06,4.98);
\draw (2.46,5.4) node[anchor=north west] {$e$};
\draw (1.84,5.44) node[anchor=north west] {$1$};
\draw (0.94,4.81) node[anchor=north west] {$2$};
\draw (0.45,3.75) node[anchor=north west] {$3$};
\draw (0.8,2.55) node[anchor=north west] {$4$};
\draw (1.73,1.8) node[anchor=north west] {$5$};
\draw (3.01,1.83) node[anchor=north west] {$6$};
\draw (3.89,2.44) node[anchor=north west] {$7$};
\draw (4.38,3.5) node[anchor=north west] {$8$};
\draw (4.05,4.62) node[anchor=north west] {$9$};
\draw [line width=1.2pt,dotted,color=qqqqff] (1.66,3.95)-- (1.16,3.01);
\draw [line width=1.2pt,dotted,color=qqqqff] (1.66,3.95)-- (2.77,3.73);
\draw [line width=1.2pt,dotted,color=qqqqff] (2.77,3.73)-- (3.32,2.66);
\draw [line width=1.2pt,dotted,color=qqqqff] (3.32,2.66)-- (2.98,1.85);
\draw [line width=1.2pt,dotted,color=qqqqff] (2.77,3.73)-- (3.06,4.79);
\draw [line width=1.2pt,dotted,color=qqqqff] (4.07,3.21)-- (3.32,2.66);
\draw (3.06,5.41) node[anchor=north west] {$10$};
\begin{scriptsize}
\fill [color=qqqqff] (1.98,5) circle (1.5pt);
\fill [color=qqqqff] (1.1,4.38) circle (1.5pt);
\fill [color=uququq] (0.75,3.36) circle (1.5pt);
\fill [color=uququq] (1.07,2.33) circle (1.5pt);
\fill [color=uququq] (1.93,1.69) circle (1.5pt);
\fill [color=uququq] (3.01,1.67) circle (1.5pt);
\fill [color=uququq] (3.89,2.29) circle (1.5pt);
\fill [color=uququq] (4.24,3.31) circle (1.5pt);
\fill [color=uququq] (3.92,4.34) circle (1.5pt);
\fill [color=uququq] (3.06,4.98) circle (1.5pt);
\fill [color=qqqqff] (1.16,3.01) circle (1.5pt);
\fill [color=qqqqff] (1.66,3.95) circle (1.5pt);
\fill [color=qqqqff] (2.77,3.73) circle (1.5pt);
\fill [color=qqqqff] (3.32,2.66) circle (1.5pt);
\fill [color=qqqqff] (2.98,1.85) circle (1.5pt);
\fill [color=qqqqff] (4.07,3.21) circle (1.5pt);
\fill [color=qqqqff] (3.06,4.79) circle (1.5pt);
\end{scriptsize}
\end{tikzpicture}
  \vspace{-15pt}
  \caption{ An outer planar graph $G$. Bold edges are edges of $T$, dashed edges are non-tree edges and dotted edges are edges of the dual spanning tree, less the dual of $e$.}
  \label{fig:disk-graph-charging}
\vspace{-10pt}
\end{figure}
\end{proof}

\begin{proof}[Proof of Lemma~\ref{lem:weak}]
  Consider an edge $\hat e \in \hat S \setminus S$. We first argue that $\hat S \setminus \{\hat e \}$ has a weak $k$-simple charging scheme.   Since $T \subseteq S \cap \hat S$, $\hat e \notin T$ and so $\hat e$ can be charged to at most once.  If $\hat e$ is in the charging path $P_{\hat S}(e)$ for another edge $e$ of $\hat S$, then we define the charging path for $e$ to be the simple path between $e$'s endpoints that is in $P_{\hat S}(e) \cup P_{\hat S}(\hat e) \setminus  \{\hat e \}$.  The resulting set of paths is a weak $k$-simple charging scheme since every edge of $P_{\hat S}(\hat e)$ is charged to one fewer time (by the removal of $\hat e$) and at most once more (by $e$).

  By induction, $S$ has a weak $k$-simple charging scheme to $T$.  Since $S$ is a greedy spanner, Equation~(\ref{eq:edge-path-ineq}) holds for every charging pair, so the weak $k$-simple charging scheme to $T$ is a  $k$-simple charging scheme to $T$.
\end{proof}

\begin{proof}[Proof of Claim~\ref{clm:embeddable-const}]
Let $V_a$ be the set of vertices in vortex $V$ that are adjacent to apex $a$.  Split vertex $a$ into two vertices $a$ and $a_V$ connected by a zero-weight edge so that $a_V$'s neighbors are $V_a \cup \{a\}$ and so that contracting the zero-weight edge gives the original graph.  Add $a_V$ to all of the bags of the path decomposition of $V$.  Now all the edges that connected $a$ to $V$ are within the vortex.  

  Consider the face in the surface-embedded part of the $\beta$-almost-embeddable graph to which $V$ is attached and let $xy$ be edge in that face that is between the first and last bags of $V$ and such that $y$ is in the first bag of $V$.  Add the edges $xa_V$ and $a_Vy$ to the embedded part of the graph and give them weight equal to the distance between their endpoints.  Now $xa_V$ is the edge in that face that is between the first and last bags of the vortex and $a$ is adjacent to a vertex that is in the surface-embedded part of the $\beta$-almost-embeddable graph.  See Figure~\ref{fig:apex-reduction}.

  The {\em splitting} of $a$ into $a_V$ increases the pathwidth of the vortex by 1.  Repeating this process for all apex-vortex pairs increases the pathwidth of each vortex by at most $\beta$.

\begin{figure}[tbh]
  \centering
    \includegraphics[height=3.0in]{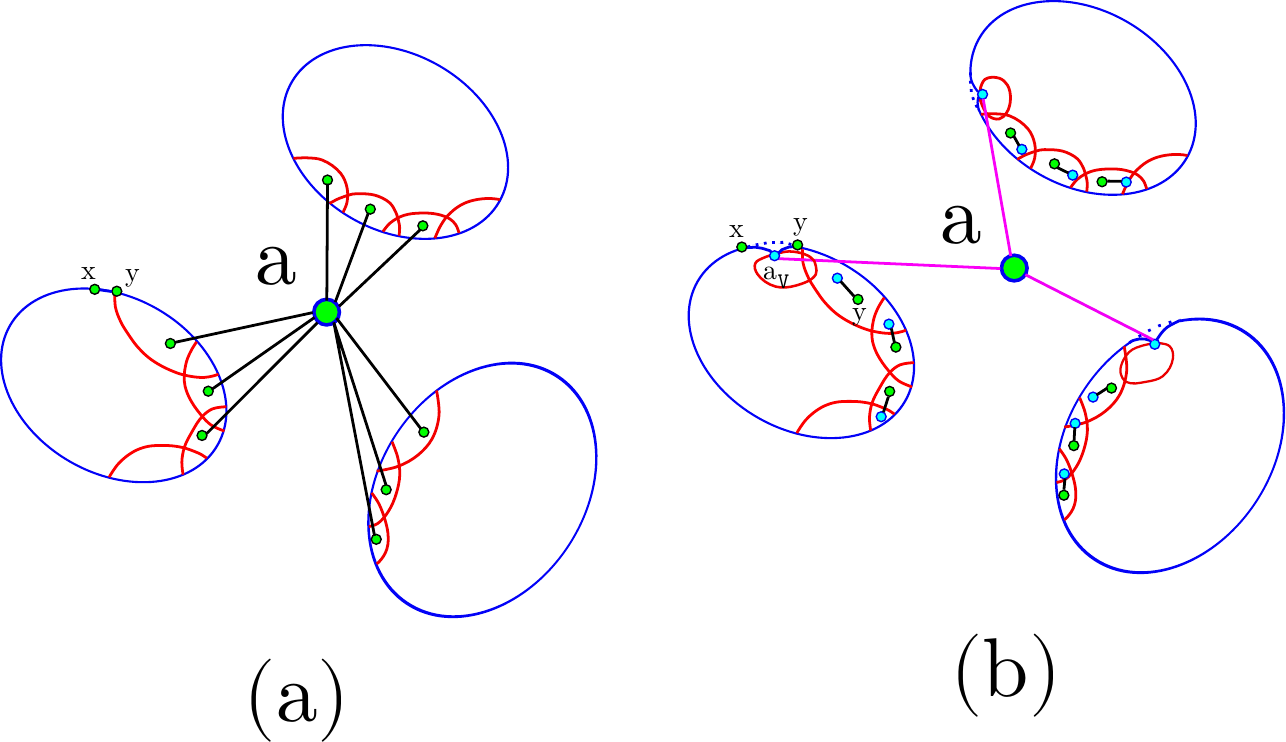}
      \caption{Apex $a$ and the vortices it is adjacent to before (a) and after (b) the reduction of Claim~\ref{clm:embeddable-const}}
       \label{fig:apex-reduction}
\end{figure}
\end{proof}

\begin{proof}[Proof of Claim~\ref{clm:init-inv}]

Note first that since there are only bold edges, the dashed-free trees of $\Phi_1$ are just the trees of $\Phi_1$. 

\subparagraph*{Invarint~(i)} Consider distinct trees $T_1$ and $T_2$ of $\MST[X_1]$ and consider $x,y \in T_1$ and $z \in T_2$.  The $\phi$-vertices $(x,z)$ and $(y,z)$ are in the same component of $H$ as witnessed by the edges of the $x$-to-$y$ path in $T_1$.  Further, $(x,z)$ cannot be connected to $(x,u)$ in $H$ where $xu \in \MST$ since that would imply $xz \in \MST$, contradicting that $T_1$ and $T_2$ are distinct trees of $\MST[X_1]$.  Therefore, there is a maximal unrooted tree of $H$ that will contain all the $\phi$-vertices of the form $(u,v)$ such that $u \in T_1$ and $v\in T_2$.

\subparagraph*{Invariant~(ii)}Consider a path $u_0,u_1, \ldots, u_k$ in $\MST[X_1]$ for $k \ge 2$; $((u_0,u_i),(u_0,u_{i+1}))$ are edges in $H$ for $i = 1,\ldots,k-1$.  Therefore $(u_0,u_1)$ (a root) and $(u_0,u_k)$ are in a common component of $H$.  Therefore any maximal tree of $H$ that contains  $(u_0,u_k)$ will contain a root.
\item[(iii)] Consider the component of $H$ described in showing $\Phi_1$ satisfies Invariant~(i).  Since $T_1$ and $T_2$ must be connected in $\MST$ by a path through ancestor nodes of the path decomposition, there is a vertex $u_i \in T_i$ such that $u_i \in X_2$ for $i = 1,2$.  Therefore $(u_1,u_2)$ is an active $\phi$-vertex in this component of $H$ and will be included in the maximal tree of this component.

\subparagraph*{Invariant~(iv)} In this case, Invariant~(iv) reduces to Invariant~(i).
\end{proof}

\begin{proof}[Proof of Claim~\ref{cl:free-inv}]
We show that $\Phi_i$ satisfies Invariants~(i),~(ii) and ~(iv) in turn:
\begin{enumerate}
\item[(i)] The only new tree of $\MST[\cup_{j \le i} X_\ell]$ compared to $\MST[\cup_{j \le i-1} X_\ell]$ is $u$.  For any pair of trees in $\MST[\cup_{\ell = j}^{i-1} X_\ell]$, Invariant~(i) holds for $\Phi_i$ because it helps for $\Phi_{i-1}$.  For a component $C$ of $\Lambda_{i-1}$, the addition of the edges $((u,w_j),(u,w_k))$ creates a new (unrooted) component spanning all $\phi$-vertices $(u,v)$ where $v$ is in the corresponding component of $\MST[\cup_{j \le i} X_j]$.  Therefore, Invariant~(i) holds for $\Phi_i$.
\item[(ii)] Since $u$ is free and no new $\phi$-vertices of the form $(u,v)$ where $u$ and $v$ are in the same tree of $\MST[\cup_{j \le i} X_j]$ are introduced, Invariant~(ii) holds for $\Phi_i$ because it holds for $\Phi_{i-1}$. 
\item[(iv)] As with $(i)$, the only new tree of $\MST[\cup_{\ell = j}^i X_\ell]$ is $u$, which only has a non-zero intersection with $X_i$.  So, for $j < i$, Invariant~(iv) holds for $\Phi_i$ because it holds for $\Phi_{i-1}$.  For $j = i$, all the components of $\MST[\cup_{\ell = j}^i X_\ell]$ are isolated vertices, so the invariant holds trivially.
\end{enumerate}
\end{proof}

\begin{proof}[Proof of Claim~\ref{cl:non-free-inv}]
We prove that $\Phi_i$ satisfies each of the invariants in turn.

\subparagraph*{Invariant~(i)} Let $T_1$ and $T_2$ be distinct trees of $\MST[\cup_{j \leq i} X_j]$.  If $T_1$ and $T_2$ are distinct trees of $\MST[\cup_{j \leq i-1} X_j]$, then the invariant holds for $\Phi_i$ because it holds for $\Phi_{i-1}$.  Otherwise, we may assume w.l.o.g.\ that $T_1$ contains $u$ and a subtree $T_3$ that is a component of $\MST[\cup_{j \leq i-1} X_j]$.  Let $x$ be a vertex of $T_3 \cap X_i$ and $w$ be a vertex of $T_2 \cap X_i$.  Then $ux$ is an edge of $\MST$ and $(u,w),(w,x)$ will be connected in $\Phi_i$ by greedy applications of the triangle rule.  By the same argument as used for showing that $\Phi_1$ satisfies Invariant~(i), $(u,w)$ will not be connected to a root $\phi$-vertex (a $\phi$-vertex corresponding to a tree edge) in $\Phi_i$.	

\subparagraph*{Invariant~(ii)} We need only prove this for the tree $T$ of $\MST[\cup_{j \leq i} X_i]$ that contains $u$ as other cases are covered by the fact that  $\Phi_{i-1}$ satisfies Invariant~(ii).  Let $T_1$ and $T_2$ be distinct trees of $\MST[\cup_{j \leq i-1} X_i]$ that are subtrees of $T$.  Let $u_1$ and $u_2$ be vertices of $T_1 \cap X_i$ and $T_2 \cap X_i$, respectively.  

We start by showing that $(u_1,u_2)$ is in a rooted dashed free tree of $\Phi_i$.  Let $v_1$ and $v_2$ be the neighbors of $u$ in $T_1$ and $T_2$, respectively.  (Note, it may be that, e.g., $u_1 = v_1$.)  By Invariant~(ii), $(v_2,u_2)$ is in a rooted dashed-free tree of $\Phi_{i-1}$.  By the triangle rule, $(v_1,u_2)$ will be connected by (a possibly non-trivial sequence of) bold edges to $(u,u_2)$ which in turn will be connected by (a possibly non-trivial sequence of) bold edges to $(v_2,u_2)$.  Therefore, $(v_1,u_2)$ will be in a rooted dashed-free tree of $\Phi_{i}$. In this next paragraph, we show that $(u_1,u_2)$  and $(v_1,u_2)$ are in the same dashed-free tree of $\Phi_{i-1}$, which we have just shown belongs to a rooted dashed-free tree of $\Phi_{i}$, showing that $(u_1,u_2)$ is in a rooted dashed free tree of $\Phi_i$.

To show that $(u_1,u_2)$  and $(v_1,u_2)$ are in the same dashed-free tree of $\Phi_{i-1}$, consider an index $k$ such that $u_1$ and $v_1$ are connected to $x_1$ in  $\MST[\cup_{\ell = k}^i X_\ell]$ and $u_2$ is connected to $x_2$ in $\MST[\cup_{\ell = k}^i X_\ell]$.  Then by Invariant~(iv) for $\Phi_{i-1}$, $(x_1,x_2)$ is in the same dashed-free tree as $(u_1,u_2)$ and $(v_1,u_2)$. 

Now consider a $\phi$-vertex $(x_1,x_2)$ where $x_k \in T_k \cap X_j$  for $k = 1,2$ and $j < i$.  Let $T_k'$ be the subtree of $T_k$ that is in $\MST[\cup_{\ell = j}^{i-1} X_\ell]$ and let $u_k = T_k' \cap X_i$ for $k = 1,2$.  We just showed that $(u_1,u_2)$ is in a rooted dashed-free tree of $\Phi_{i}$.  By Invariant~(iv) for $\Phi_{i-1}$, $(x_1,x_2)$ and $(u_1,u_2)$ are in the same dashed-free tree.  Therefore, $(x_1,x_2)$ is in a rooted dashed-free tree of $\Phi_{i}$. 

\subparagraph*{Invariant~(iv)} Consider $j \leq i$, trees $T_1, T_2$ of $\MST[\cup_{\ell = j}^i X_\ell]$, and $\phi$-vertices $(y_{j,1},y_{j,2})$ and $(y_{i,1},y_{i,2})$  as defined in Invariant~(iv).  For the case $j = i$, the proof that $\Phi_{i}$ satisfies Invariant~(iv) is the same as for $\Phi_1$.  Further, if $u \notin T_1,T_2$, $\Phi_{i}$ satisfies Invariant~(iv) because $\Phi_{i-1}$ satisfies Invariant~(iv), therefore, we assume w.l.o.g.\ that $u \in T_1$.  Finally, consider the components of $T_1$ in $\MST[\cup_{\ell = j}^{i-1} X_\ell]$; if $y_{j,1}$ and $y_{i,1}$ are in the same component, then $\Phi_{i}$ satisfies Invariant~(iv) because $\Phi_{i-1}$ satisfies Invariant~(iv).  Therefore, we assume they are in different components, $T_1^p$ and $T_1^q$, respectively.  

Let $x$ be the neighbor of $u$ in $T_1^p$.  By Invariant~(iv) for $\Phi_{i-1}$, $(x,y_{i,2})$ and $(y_{j,1},y_{j,2})$ are in a common unrooted dashed-free tree.  We show that $(x,y_{i,2})$ and $(y_{i,1},y_{i,2})$ are in a common unrooted dashed-free tree of $\Phi_i$, proving this invariant is held.  Let $y$ be the neighbor of $u$ in $T_1^q$.
\begin{enumerate}[nolistsep,noitemsep]
\item By the triangle rule, $(x,y_{i,2})$ will get connected by (a possibly non-trivial sequence of) bold edges to $(u,y_{i,2})$ because $xu \in \MST$ and $(u,y_{i,2})$ will get connected by (a possibly non-trivial sequence of) bold edges to $(y, y_{i,2})$ because $yu \in \MST$.
\item Let $k$ be the index such that $y$ and $y_{i,1}$ are connected to $a$ in $\MST[\cup_{\ell = k}^i X_\ell]$ and $y_{i,2}$ is connected to $b$ in $\MST[\cup_{\ell = k}^i X_\ell]$.  Then by Invariant~(iv), $(a,b)$ is in the same unrooted dashed-free tree as $(y, y_{i,2})$ and $(y_{i,1},y_{i,2})$.
\end{enumerate}
Together these connections show that $(x,y_{i,2})$ and $(y_{i,1},y_{i,2})$ are in a common unrooted dashed-free tree.
\end{proof}

\begin{proof}[Proof of Claim~\ref{cl:lamda-dist-rank}]
We prove the claim for each case of the construction of $\Lambda_i$:

\begin{enumerate}[nolistsep,noitemsep]
\item If $u = v$, we only add new edges to $\Lambda_{i-1}$ to obtain $\Lambda_{i}$. Since each added edge has distinct rank that is larger than the ranks of edges of $\Lambda_{i-1}$, the claim follows.
\item If $u\not= v$ and $uv \in \MST[X_i]$, since $r_i(xy) = r_{i-1}(xy)$ for any edge $xy$ such that $x,y \neq u$, we only need to consider the case when $u \in \{x,y\}$. Observe that for each neighbor $w$ of $u$,  $r_i(uw) = r_{i-1}(vw)$ if $vw\in \Lambda_{i-1}$ or $uw$ is among the edges that are added to $\Lambda_{i-1}$. Since each newly added edge has unique rank, the claim follows. 
\item Otherwise, we have $r_{i}(v_jv_{j+1}) = r_{i-1}(v)v_j$ for $1 \leq j \leq p-1$ and $r_i(xy) = r_{i-1}(xy)$ for all other edges. Thus, each edge in $\Lambda_i$ has unique rank by the induction hypothesis. 
\end{enumerate}
\end{proof}


Let $w_k$ ($k=1,2$) be a vertex in $ N^{r}_{i-1}(v_k) \cup \{v_k\}$ such that:

\begin{proof}[Proof of Lemma~\ref{lm:associated-pseudo-triangle}]
We investigate how dashed edges are changed into mixed edges as this is when a charging pseudo-triangle arises. Recall that dashed edges are added to $\Phi$ when we process free vertices. Let $u$ be a free vertex that is introduced in bag $X_i$ and let $((u,v_1),(u,v_2))$ be a mixed edge (that is, a dashed edge in $\Phi_i$ that is later converted to a mixed edge).  Then by our construction, $v_1$ and $v_2$ are in the same component of $\MST[\cup_{k \le i-1}X_k]$. For each $j \geq  i$ and $w \in X_i$, let $T^{ij}_w$ be the component of $ \MST[\cup_{i \leq k \leq j}X_k]$ containing $w$. We will say that an introduced vertex of a bag $X_j$ is {\em branching} if it is not free, is not the forgotten vertex of $X_j$ and is not connected to the forgotten vertex of $X_j$ by $\MST$.  Note that dashed edges are converted to mixed edges when processing free and branching vertices.

 We say that a tree $T^{ij}_w \in  \MST[\cup_{i \leq k \leq j}X_k]$ is \emph{forgotten} in $X_j$ if $j = |I|$ or :
	\begin{itemize}[nolistsep,noitemsep]
	\item $T_w^{ij} \cap X_j = \{v\}$ which is the forgotten vertex of $X_j$.
	\item the introduced vertex of $X_j$ is branching or free.
	\end{itemize}
We note that at most one tree can be forgotten in each bag $X_j$. 

For each vertex $x \in \Lambda_{i-1}$, let $N_{i-1}^r(x)$ be the set of vertices in $\Lambda_{i-1}$ reachable from $x$ via paths consisting of edges of ranks larger than $r$ in $\Lambda_{i-1}$. Let $F_{x,r}^{ij} = \cup_{ w\in N_{i-1}^r(x) \cup \{x\}}T^{ij}_{w}$. We say that the forest $F_{x,r}^{ij}$ is forgotten in $X_j$ if every tree in $F_{x,r}^{ij}$ is forgotten in $X_{j'}$ for some $j ' \leq j$ and $F_{x,r}^{ij} \cap X_j \not= \emptyset$.

\begin{claim} \label{clm:realized-dashed-edges}
Let $r$ be the rank of the edge $v_1v_2 \in \Lambda_{i-1}$. If the dashed edge $((u,v_1),(u,v_2))$ is changed into a mixed edge, there exists a bag $X_j$ for $j \geq i$ such that $T^{ij}_u \cap \{F^{ij}_{v_1,r} \cup F^{ij}_{v_2,r}\} = \emptyset$, $T^{ij}_{u} \cap X_j \not= \emptyset$ and exactly one of the forests  $F_{v_1,r}^{ij}, F_{v_2,r}^{ij}$ is forgotten in $X_j$.
\end{claim}
\begin{proof}
Suppose that the claim fails, then there are two cases:

	\begin{enumerate}[nolistsep,noitemsep]
	\item  There exists $j \geq i$ such that $T^{ij}_u = T_{w}^{ij}$ for some $w \in F^{ij}_{v_1,r}$ and $F^{ij}_{v_2,r} \cap X_j \not = \emptyset$. In this case, $i < j$ since $u$ is a free vertex. Let $p$ be the index such that $u$ and $w$ are connected to $x$ in $\MST[\cup_{i \leq q \leq p} X_q]$. 
Let $w_k$ ($k=1,2$) be a vertex in $ N^{r}_{i-1}(v_k) \cup \{v_k\}$ such that:
	\begin{enumerate}[nolistsep,noitemsep]
	\item[(i)] $T^{ij}_{w_k} \cap X_p \not= \emptyset$
	\item[(ii)] Subject to (i), the distance $d_{\Lambda_{i-1}}(w_k,v_k)$ is minimum. 
	\end{enumerate}
Let $y_k = T^{ij}_{w_k} \cap X_p$. Since the tie-breaking rule prefers changing the dashed edges of higher ranks into mixed edges, $(u,v_k)$ is in the same dashed-free tree of $\Phi_p$ as $(u,w_k)$. By Invariant~(iv), $(u,w_k)$ and $(w,w_k)$  are in the same dashed-free tree of $\Phi_p$ as $(x,y_k)$. Since $(w,w_k)$ is in a rooted dashed-free tree of $\Phi_p$, by Invariant~(i), $(u,v_k)$ is in a rooted dashed free tree. Therefore, the dashed edge  $((u,v_1),(u,v_2))$ is not converted to a mixed edge.
	
	\item There exists $j \geq i$ such that $T^{ij}_u$ is forgotten and both forests  $F^{ij}_{v_1,r}, F^{ij}_{v_2,r}$ are not forgotten in $X_j$. Since $u$ is a free vertex, there must be a vertex $\hat{u}$ in $X_i$ such that $\hat{u} \in T^{ij}_u$. Let $p$ be the index such that $u$ and $\hat{u}$ are connected to $x$ in $\MST[\cup_{i \leq q \leq p} X_q]$. Let $w_k$ ($k=1,2$) be a vertex in $ N_{i-1}^{r}(v_k) \cup \{v_k\}$ as in the first case and $y_k = T^{ij}_{w_k} \cap X_p$. Then, by the tie-breaking rule, $(u,v_k)$ and $(u,w_k)$ are in the same dashed-free tree of $\Phi_p$. By Invariant~(iv), $(u,w_k)$ and $(\hat{u},w_k)$ are in the same dashed-free tree of $\Phi_p$ as  $(x,y_k)$. Therefore, there is a cycle in which $((u,v_1),(u,v_2))$ is the most recent added dashed edge. Thus,  $((u,v_1),(u,v_2))$ is deleted and not converted to a mixed edge.\qedhere
\end{enumerate}
\end{proof}
We now bound the number of pseudo-triangles that contain an edge $e$ of $\MST(G)$.  Let $X_i$ be the bag that containing $e$ and $\Lambda_i$ be the corresponding contracted forest.  Let $u_0,v_0$ be two vertices of $\Lambda_i$ in the same tree of $\Lambda_i$ such that $e$ is in the $u_0$-to-$v_0$ path $P_\MST(u_0,v_0)$. We say that a pseudo-triangle \emph{strongly contains} $P_\MST(u_0,v_0)$ if $P_\MST(u_0,v_0)$ is a subpath of the pseudo-edge and the rank of every edge of the psedudo-edge of the triangle is at least the minimum rank over edges of $P_\MST(u_0,v_0)$. 

\begin{claim}\label{clm:num-nesting}
There are at most $\pw$ pseudo-triangles strongly containing $P_\MST(u_0,v_0)$. 
\end{claim} 
\begin{proof}
Let $\Delta_1, \ldots, \Delta_q $  be pseudo-triangles that strongly contain $P_\MST(u_0,v_0)$. Let $\Delta_k = \{u_k,v_k,w_k\}$ and $X_{s_k}$ be the bag that has $w_k$ as the introduced vertex ($1 \leq k \leq q$). Let $r$ be the minimum rank over edges of $P_\MST(u_0,v_0)$ and $X_\ell$ be the bag in which one of the forests $F^{i\ell}_{u_0,r}, F^{i\ell}_{v_0,r}$, say $F^{i\ell}_{u_0,r}$, is forgotten. Then, $u_k \in F^{i\ell}_{u_0,r}$ and $v_k \in F^{i\ell,r}_{v_0}$. We can assume w.l.o.g\ that $s_1 \leq s_2 \leq \ldots \leq s_q$.  Then, $F_{u_k,r}^{s_k\ell} = F^{i\ell}_{u_0,r} \cap \MST[\cup_{s_k \leq j \leq \ell}X_j]$ and $F_{v_k,r}^{s_k\ell} = F^{i\ell}_{v_0,r} \cap \MST[\cup_{s_k \leq j \leq \ell}X_j]$. By Claim~\ref{clm:realized-dashed-edges}, $T^{s_k\ell}_{w_k} \cap X_{\ell} \not= \emptyset$. Furthermore, all the tree $T^{s_k\ell}_{w_k}$ must be disjoint since otherwise, say $T^{s_j\ell}_{w_j}$ is a subtree of $T^{s_k\ell}_{w_k}$ $( k \leq j)$,  the second case of the proof of Claim~\ref{clm:realized-dashed-edges} implies that the dashed edge $((w_j,u_j),(w_j,v_j))$ is removed from $\Phi$. Hence, $q \leq \pw$.
\end{proof}
Let $e$ be an arbitrary edge of $\MST(G)$. Let $X_i$ be the bag that containing $e$ and $\Lambda_i$ be the corresponding contracted forest. By the way we build the contracted forest, there are at most two edges, say $u_0v_0$ and $v_0w_0$, of $\Lambda_i$ that are incident to the same vertex $v_0$ such that $e \in P_\MST(u_0,v_0) \cap P_\MST(v_0,w_0)$. We observe that any pseudo-triangle that contains $e$ in the pseudo-edge must contain the path $P_\MST(u,v)$ for some $u,v \in \Lambda_i$ such that one of two edges $u_0v_0$, $v_0w_0$ is in the path $P_{\Lambda_{i-1}}(u,v)$. Let $\Gamma(u_0,v_0)$ be the set of pseudo-triangles that have $P_{\Lambda_{i-1}}(u,v)$ containing $u_0v_0$. We define $\Gamma(v_0,w_0)$ similarly. We will show that $|\Gamma(u_0,v_0)|,|\Gamma(u_0,w_0)| \leq \pw^2$, thereby, proving the lemma. 

We only need to show that $|\Gamma(u_0,v_0)| \leq \pw^2$ since $|\Gamma(u_0,w_0)| \leq \pw^2$ can be proved similarly.  Let $\Delta_1, \Delta_2, \ldots, \Delta_p$ be the pseudo-triangles containing $e$ in the pseudo-forest such that for each $k$:
	\begin{enumerate}[nolistsep,noitemsep]
	\item Each triangle $\Delta_k$ contains distinct path $P_\MST(u_k,v_k)$ as subpath in the pseudo-edge. 
	\item Edge $u_0v_0$ is in the path $P_{\Lambda_{i-1}}(u_k,v_k)$ of $\Lambda_{i-1}$.
	\item Each path $P_\MST(u_k,v_k)$ has distinct rank.
	\end{enumerate}	 
Then, by the way we construct the contracted forest, if the minimum rank over edges of $P_\MST(u_j,v_j)$ is smaller than the minimum rank over edges of $P_\MST(u_k,v_k)$, then $P_{\Lambda_i}(u_k,v_k) \subsetneq P_{\Lambda_i}(u_j,v_j)$ for $1 \leq j \not= k \leq p$.  Therefore, we can rearrange $\Delta_1, \Delta_2, \ldots, \Delta_k$ such that $P_{\Lambda_{i}}(u_1,v_1)\subsetneq P_{\Lambda_{i}}(u_2,v_2) \subsetneq \ldots \subsetneq P_{\Lambda_{i}}(u_p,v_p)$. Thus, $p \leq \pw$ . By Claim~\ref{clm:num-nesting}, for each $k$, there are at most $\pw$ pseudo-triangles containing the same subpath $P_\MST(u_k,v_k)$. Therefore, $|\Gamma(u_0,v_0)| \leq \pw^2$. 
\end{proof}
\section{Figures} \label{app:figures}

\begin{figure}[hb] 
  \centering
  \vspace{-20pt}
  \input{figs/cutting-surface.tex}
  \vspace{-15pt}
  \caption{The surface $(2)$ is obtained by cutting the surface $\Sigma$ in $(1)$ along $T_{\widehat{G}} \cup X$. Small ovals are cycles $C_1,C_2,\ldots, C_\beta$.}
  \label{fig:cutting-surface}
\end{figure}

\begin{figure}[tbh]
  \centering
    \includegraphics[height=2.4in]{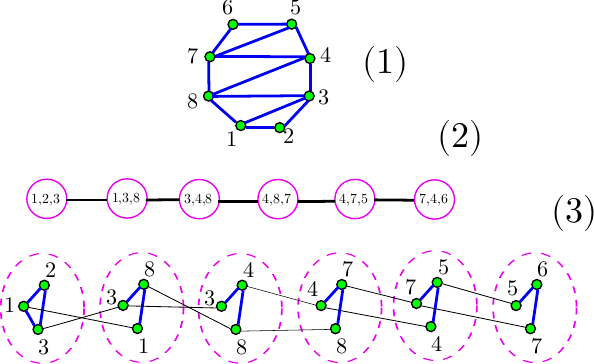}
      \caption{The normalized graph (3) of a graph (1) with path decomposition (2). Bold edges in (3) are edges of the original graph and thin edges in (3) are zero weighted edges.}
       \label{fig:normalized-graph}
\end{figure}

\begin{figure}[tbh]
  \centering
    \includegraphics[height=6.0in]{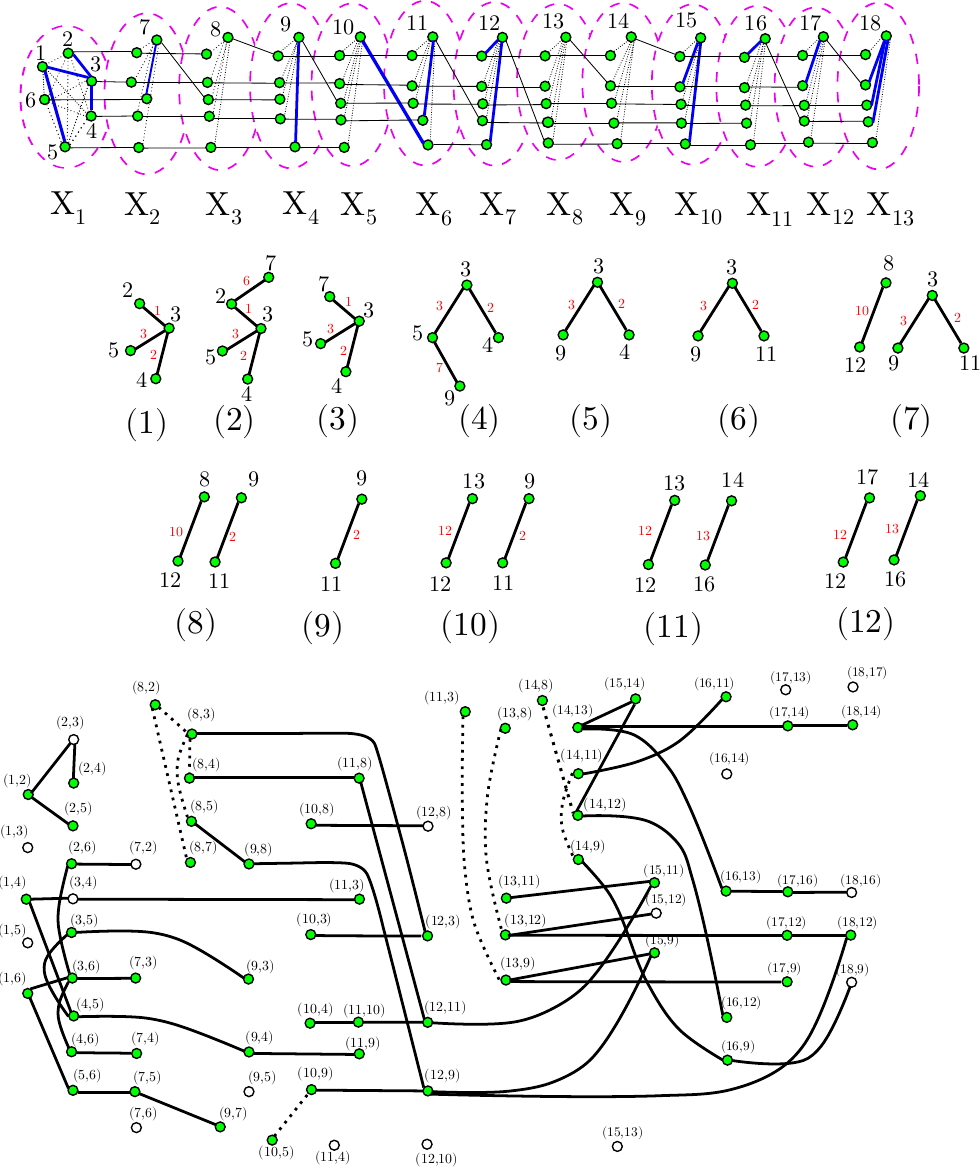}
      \caption{Top: A normalized graph and its MST. Dotted edges are non-tree edges. Non-dotted edges are edges of the MST. Thin edges are zero weighted edges. Center: Contracted forest $\Lambda_1, \ldots, \Lambda_{12}$. Small numbers are ranks of the edges of the contracted forests. Bottom: Charging forest for the graph. Dotted edges are mixed edges. Hollow vertices are roots of the charging forest}
       \label{fig:charging-mst}
\end{figure}

\fi
\end{document}